\documentclass[journal]{IEEEtran}

\usepackage{mathtools,bm}
\usepackage{amssymb}
\usepackage{booktabs}
\usepackage{amsthm}
\usepackage{algorithm}
\usepackage{algpseudocode}
\usepackage[capitalize,noabbrev]{cleveref}

\theoremstyle{plain}
\newtheorem{proposition}{Proposition}
\newtheorem{lemma}{Lemma}
\newtheorem{corollary}{Corollary}

\theoremstyle{definition}
\newtheorem{definition}{Definition}

\theoremstyle{remark}
\newtheorem{remark}{Remark}

\begin{document}

\title{Designed-Source Reductions and a Dual-Purpose Feasibility Band for Semantic Rate-Distortion}

\author{Joss~Armstrong\thanks{The author is with Ericsson Research, Athlone, Ireland (e-mail: joss.armstrong@ericsson.com). ORCID: 0009-0009-3462-9679.}}

\markboth{Preprint}{Armstrong: Designed-Source Reductions and a Dual-Purpose Feasibility Band}

\maketitle

\begin{abstract}
The joint rate-distortion framework of Stavrou and Kountouris
(SK)~\cite{sk2023} characterises dual-fidelity tradeoffs for semantic
communication on stochastic semantic sources. Many task-oriented
communication systems instead use designed sources, where the
semantic object is a deterministic oracle allocation $\phi^*(t)$
rather than a stochastic quantity given by nature. We isolate the
subclass of designed sources under smooth concave utility with
assumptions A1, A2 and Euclidean allocation codomain, and restrict
the encoder class to deterministic common-category mappings. Within
this subclass the SK exponential-tilting decoder and generalised
Blahut--Arimoto iteration specialise to conditional-mean decoding
and Lloyd--Max stationarity on $\phi^*(t)$. When the second fidelity
is a monotone single-letter distortion, the joint problem stays
inside the SK admissible class; the common-category SK rate is
lower-bounded by the max of the corresponding Shannon rate-distortion
functions, with equality only when the common-category reconstruction
is compatible and RDF-optimal. When the second
fidelity is aggregate verification, the joint problem leaves the SK
single-letter class and admits a constrained-design feasibility band
$R_{\min}(\varepsilon^*) \leq R \leq R_{\max}(\beta^*)$ of width
$\log_2(K_{\max}/K_{\min})$ bits in partition cardinality. The
reduction and the band are scope statements on the SK apparatus, not
modifications to it. A smart-grid economic-dispatch example with a
non-technical-loss-detection contrast illustrates the band.
\end{abstract}


\section{Introduction}\label{sec:intro}

Stavrou and Kountouris~\cite{sk2023} characterise the joint
rate-distortion tradeoff between a semantic
distortion $D_s$ and an observation distortion $D_o$ for goal-oriented
semantic communication, an information-theoretic instance of the
broader task-oriented communications agenda surveyed by G\"{u}nd\"{u}z
\emph{et al.}~\cite{gunduz2023}. Their analysis develops an exponential-tilting
parametric decoder, a generalised Blahut-Arimoto iteration coupling the
two reconstruction marginals, and a pair of Markov-chain hypotheses
under which the joint rate-distortion function admits the closed form
$R(D_s, D_o) = \max\{R(D_s), R(D_o)\}$. The machinery is necessary for
stochastic semantic sources, where the semantic object is given by
nature rather than chosen by the designer.

A substantial portion of the comms-applied workload motivating
goal-oriented quantisation does not match that hypothesis. The
resource-allocation and power-scheduling mechanisms
of~\cite{zou2023,sun2024} choose a deterministic allocation function
under smooth concave utility, so the semantic object is an oracle
action $\phi^*(t)$ that is a known function of the observation. When
such a mechanism faces a second fidelity requirement, for instance
meter-level fraud detection~\cite{leite2018,messinis2018,buzau2019}
alongside dispatch efficiency in a smart-grid distribution network,
the SK two-distortion
characterisation is the natural theoretical home. This paper analyses
that subclass.

The result is a structural theorem about which subclass of
two-distortion problems sits inside the SK admissible class and which
subclass leaves it. We pin down a subclass (designed sources under A1,
A2, Euclidean allocation codomain) in which the SK exponential-tilting
decoder and the generalised Blahut--Arimoto iteration reduce to
textbook objects~\cite{shlezinger2019,strinati2021,zou2023,sun2024}, and
the Markov-chain hypotheses become automatic under the common-category
injectivity condition, and we
identify the regime within that subclass in which the SK
max-decomposition $R(D_s, D_o) = \max\{R(D_s), R(D_o)\}$ does not
apply because the joint design problem leaves the SK admissible
class. The four propositions and two corollaries instantiate that
scope statement.

Within the
designed-source subclass (\Cref{sec:designed-source}: assumptions
A1--A2 on the utility $U$, deterministic
$\phi^*(t) = \arg\max_r U(t, r)$),
\Cref{prop:task-loss,prop:decoder,prop:encoder} (\Cref{sec:reduction})
reduce the SK machinery to short statements, all stated on the
deterministic common-category encoder class of \Cref{sec:mapping}:
SK's exponential-tilting parametric solution becomes the conditional
mean on the squared-oracle-error surrogate (exact for isotropic
quadratic utility), SK's generalised Blahut-Arimoto specialises to
Lloyd-Max on the oracle action under this restricted architecture,
and the task loss is sandwiched by the within-category variance up to
A1/A2 curvature constants. \Cref{prop:opposing} (\Cref{sec:opposing}) identifies a complementary
regime: the two design targets \emph{oppose}
each other when a single $K$-way partition is load-bearing for both
allocation and aggregate verification, so the joint tradeoff is a
constrained-design band
$R_{\min}(\varepsilon^*) \leq R \leq R_{\max}(\beta^*)$
(equivalently $K_{\min} \leq K \leq K_{\max}$ under uniform assignment)
rather than $\max\{R(D_s), R(D_o)\}$.
Opposing monotonicity arises because finer categories improve
coordination while degrading verification, a structure generic to
designed mechanisms whose system-level partition is load-bearing for
both coordination and verification.
\Cref{cor:no-gap,cor:rate-gap} pin down the relationship to the SK
rate along a monotonicity axis. When the verification-side mapping
$R \mapsto D_o(R)$ is monotone decreasing in rate, as a standard
rate-distortion function requires, \Cref{cor:no-gap} specialises the
SK characterisation under designed-source mapping to the max of two
standard Shannon rate-distortion functions, on $\phi^*(T)$ and on $T$
respectively, under compatible common-category reconstructions
(\Cref{cor:no-gap}); the inequality $\geq$ holds in general, and in
the squared-oracle-error case used in the numerical illustrations the
two restricted deterministic marginal problems coincide, while equality
with the unrestricted Shannon RDFs additionally requires the common
finite-$K$ encoder to be RDF-optimal. When the verification-side mapping is monotone increasing in
rate, as it is under aggregate testing with pool-size-driven detection,
the joint problem sits outside the SK admissible class;
\Cref{cor:rate-gap} replaces the max-decomposition with the
constrained-design band of \Cref{prop:opposing}, of width
$R_{\max} - R_{\min} = \log_2(K_{\max}/K_{\min})$ bits when non-empty.

\Cref{sec:lzp} compares the reduction with the
Gaussian-semantic-source reverse-water-filling result of Liu, Shao,
Zhang and Poor~\cite{lszp2022}. \Cref{sec:numerical} runs a comms-side
worked-example pair: smart-grid economic dispatch on the
designed-source side illustrates that the two SK distortions
trace a single $R(\varepsilon)$ curve up to the
\Cref{prop:task-loss} sandwich, and a non-technical-loss-detection
contrast~\cite{leite2018} on an unstructured source shows that no such
reduction occurs when no oracle action exists. \Cref{sec:ib-relation}
positions the reduction against the information-bottleneck
literature~\cite{tishby1999ib,strouse2017dib,goldfeld2020ib} and the
indirect-rate-distortion line~\cite{dobrushin1962,wolf1970,witsenhausen1980,lszp2022}: the
designed-source setting maps to deterministic IB with the $K$-means
limit~\cite{strouse2019ib} instantiated on the transformed samples
$\{\phi^*(t_i)\}$ rather than the raw types $\{t_i\}$, and
Witsenhausen's indirect-rate-distortion construction underlies
\Cref{cor:no-gap}. \Cref{sec:scope} delimits the reduction: general
stochastic-source SK, adversarial or strategic SK variants, non-smooth
or non-concave utility, and non-Euclidean allocation codomains all
remain outside scope and require the full SK machinery or a different
reduction.

\paragraph{Mechanism-design context.}
We do not develop the mechanism-design implications here; the result is
stated at the level of the rate-distortion characterisation, and the
oracle allocation $\phi^*$ is treated as given.

\section{Preliminaries}\label{sec:prelim}

\subsection{The SK Framework}\label{sec:sk-recap}

A memoryless pair $(x, z)$ has joint distribution $p(x, z)$ over finite
alphabets $\mathcal{X} \times \mathcal{Z}$, with reconstruction
alphabets $\hat{\mathcal{X}}, \hat{\mathcal{Z}}$. Here $x$ is the
semantic source and $z$ is the observation at the encoder. Two
distortion measures $d_s$ (semantic) and $d_o$ (observation) yield the
joint rate-distortion function~\cite[Lemma~1]{sk2023}:
\begin{equation}\label{eq:sk-rd}
R(D_s, D_o) = \inf_{\substack{q(\hat{z}, \hat{x} \mid z):\;
  \mathbb{E}[d_s] \leq D_s,\;
  \mathbb{E}[d_o] \leq D_o}} I(z;\, \hat{z}, \hat{x}).
\end{equation}
SK characterise~\eqref{eq:sk-rd} through three results that the
designed-source reduction will repeatedly contrast against.

\paragraph{Max-decomposition under Markov-chain hypotheses.}
By the chain rule of mutual information,
$I(z; \hat{z}, \hat{x}) = I(z; \hat{x}) + I(z; \hat{z} \mid \hat{x})
 = I(z; \hat{z}) + I(z; \hat{x} \mid \hat{z})$,
so non-negativity of conditional mutual information gives the lower
bound $I(z; \hat{z}, \hat{x}) \geq \max\{I(z; \hat{x}), I(z; \hat{z})\}$.
SK's Lemma~2 records that this bound is tight,
\begin{equation}\label{eq:sk-max}
R(D_s, D_o) \;=\; \max\{R(D_s),\, R(D_o)\},
\end{equation}
if and only if the two Markov chains
\begin{equation}\label{eq:sk-markov}
z - \hat{x} - \hat{z}
\qquad\text{and}\qquad
z - \hat{z} - \hat{x}
\end{equation}
hold concurrently~\cite[Lemma~2]{sk2023}. The individual
rate-distortion functions in~\eqref{eq:sk-max} are
$R(D_s) = \min_{q(\hat{x} \mid z):\, \mathbb{E}[d_s] \leq D_s}
  I(z; \hat{x})$
and
$R(D_o) = \min_{q(\hat{z} \mid z):\, \mathbb{E}[d_o] \leq D_o}
  I(z; \hat{z})$.
The hypotheses in~\eqref{eq:sk-markov} are non-trivial. They require
that the reconstruction pair $(\hat{z}, \hat{x})$ encodes $z$ through
either component alone.

\paragraph{Exponential-tilting parametric optimiser.}
Let $s_1, s_2 \leq 0$ be the Lagrange multipliers associated with the
two distortion constraints, and let $\nu^*(\hat{z}, \hat{x})$ be the
joint output marginal at the optimum. SK's Theorem~1 gives the implicit
parametric form of the conditional optimiser~\cite[Theorem~1]{sk2023}:
\begin{equation}\label{eq:sk-tilting}
q^*(\hat{z}, \hat{x} \mid z)
\;=\;
\frac{e^{s_1 d_s(z, \hat{x}) + s_2 d_o(z, \hat{z})}\,
      \nu^*(\hat{z}, \hat{x})}
     {\lambda(z)},
\end{equation}
where
$\lambda(z) = \sum_{\hat{z}, \hat{x}}
  e^{s_1 d_s(z, \hat{x}) + s_2 d_o(z, \hat{z})}\,
  \nu^*(\hat{z}, \hat{x})$.
The rate at the optimum is
$R(D_s^*, D_o^*) = s_1 D_s^* + s_2 D_o^*
 - \sum_z p(z) \log \lambda(z)$.
The two multipliers and the joint marginal $\nu^*$ are coupled through
the exponential tilt over the pair $(\hat{z}, \hat{x})$, so neither
component can be solved in isolation.

\paragraph{Generalised Blahut-Arimoto iteration.}
The optimiser~\eqref{eq:sk-tilting} is computed by an alternating
update on the joint output marginal. Writing
$A(z, \hat{z}, \hat{x}) = e^{s_1 d_s(z, \hat{x}) + s_2 d_o(z, \hat{z})}$,
SK's Theorem~4 gives the iteration~\cite[Theorem~4]{sk2023}
\begin{equation}\label{eq:sk-gba}
\nu^{(k+1)}(\hat{z}, \hat{x})
= \nu^{(k)}(\hat{z}, \hat{x})
\sum_{z} \frac{p(z)\, A(z, \hat{z}, \hat{x})}
              {\sum_{\hat{z}', \hat{x}'} A(z, \hat{z}', \hat{x}')\,
               \nu^{(k)}(\hat{z}', \hat{x}')}.
\end{equation}
The iteration is over the joint $(\hat{z}, \hat{x})$ marginal because
the exponential weight $A$ couples the two distortions; it does not
decompose into separate updates on $\nu(\hat{z})$ and $\nu(\hat{x})$.

\paragraph{Continuous and finite alphabets.}
SK state the finite-alphabet formulation. The continuous examples used
below admit two readings. Either they are standard Polish-space
extensions of rate-distortion theory under measurable distortions, or
they are empirical finite-sample quantisations of the same problem.
The formal SK comparison in \Cref{cor:no-gap} is stated under the
Polish-space measurability assumptions of \Cref{lem:indirect}.

\subsection{Designed-Source Mechanism Class}\label{sec:designed-source}

A coordinator manages a resource pool on behalf of agents with demand
types $t \sim F$ over $\mathcal{T} \subseteq \mathbb{R}^d$. The utility
$U: \mathcal{T} \times \mathbb{R}^r \to \mathbb{R}$ satisfies:
\begin{itemize}
\item[\textbf{A1.}] $\beta_U$-smoothness:
  $\nabla^2_{rr} U(t, r) \succeq -\beta_U I$ for all $(t, r)$, with
  $\beta_U > 0$.
\item[\textbf{A2.}] $\alpha$-strong concavity:
  $\nabla^2_{rr} U(t, r) \preceq -\alpha I$ for all $(t, r)$, with
  $0 < \alpha \leq \beta_U$.
\end{itemize}
The constants are positive (a non-degenerate curvature sandwich
requires both bounds active) and $\alpha \leq \beta_U$ is automatic,
since $-\beta_U I \preceq \nabla^2_{rr} U \preceq -\alpha I$ forces
$-\beta_U \leq -\alpha$. The action codomain $\mathbb{R}^r$ is an
unconstrained convex set, so the oracle optimum
$\phi^*(t) = \arg\max_r U(t, r)$ is interior and the first-order
condition $\nabla_r U(t, \phi^*(t)) = 0$ holds. Smooth concavity is
the standard modelling assumption for resource allocation in convex
analysis~\cite{boyd2004}; together with the oracle action $\phi^*$
defined below, A1 and A2 specify the designed-source mechanism class
studied in this paper. The
subscript on $\beta_U$ distinguishes the curvature constant from the
detection-power symbols $\beta_{\mathrm{agg}}, \beta^*$ introduced in
\Cref{sec:opposing}.

The oracle allocation is $\phi^*(t) = \arg\max_r U(t, r)$. A category
signal $c: \mathcal{T} \to \{1, \ldots, K\}$ induces the category
allocation $\phi(k) = \mathbb{E}[\phi^*(t) \mid c(t) = k]$ and
within-category variance
$\varepsilon = \mathbb{E}[\|\phi^*(t) - \phi(c(t))\|^2]$.

\subsection{Variable Mapping}\label{sec:mapping}

The SK observation $z$ maps to the demand type $t$. The semantic source
$x$ maps to the oracle allocation $\phi^*(t)$, which is deterministic in
$t$. The semantic distortion $d_s$ maps to welfare loss, and the
observation distortion $d_o$ maps to detection degradation (the
pool-size-driven aggregate-test loss developed in \Cref{sec:opposing}).
We distinguish two regimes for this $d_o$ mapping. In the SK-admissible
monotone regime of \Cref{cor:no-gap}, $d_o$ is a single-letter
observation distortion on $T$. In the aggregate-verification regime of
\Cref{sec:opposing}, the detection loss is instead a pool-level side
constraint, not an SK single-letter distortion; we retain the SK
notation as a point of comparison rather than as a claim that the joint
problem stays inside the SK admissible class.
The reconstruction $\hat{x}$ maps to the category allocation
$\phi(c)$, and the rate $R$ maps to the information budget
$R_H := H(C) \leq R_K := \log_2 K$. The observation-side reconstruction is the category
index itself, $\hat{z} = c(t)$, which the aggregate test uses to pool
meters of size $m_k = M/K$ within category $k$. Both $\hat{x} = \phi(c)$
and $\hat{z} = c$ are deterministic functions of $c$, with $\hat{x}$
determined by $\hat{z}$ via $\hat{x} = \phi(\hat{z})$. When the
category-allocation map $\phi$ is injective, $\hat{z}$ is in turn
determined by $\hat{x}$ via $\hat{z} = \phi^{-1}(\hat{x})$; this
non-degeneracy holds generically for continuous $T$ with Lloyd--Max
categorisation and is assumed throughout.

The structural specialisation is that $x = \phi^*(z)$ is a known
deterministic function, not a stochastic source. This is natural when
the designer chooses $\phi^*$.

Three rate quantities recur and should be kept distinct. The
partition-cardinality rate is $R_K := \log_2 K$, the number of bits
required to index a $K$-way partition. The entropy rate is
$R_H := H(C)$, the entropy of the category indicator under the source
distribution. The Shannon rate-distortion function is
$R_{\mathrm{Sh}}(D)$, the information-theoretic floor at distortion
target $D$ on the underlying source. For a category representation $C$
whose reconstruction attains distortion $D$, these satisfy
$R_{\mathrm{Sh}}(D) \leq R_H \leq R_K$, with $R_H = R_K$ iff $C$ is
uniform on $\{1, \ldots, K\}$ and $R_{\mathrm{Sh}}(D) = R_H$ when the
induced category representation attains the Shannon RDF at distortion
$D$. We retain the bare symbol $R$ for the SK rate variable in
equations where the surrounding text fixes the intended specialisation:
lower endpoints of cardinality-level outer bands use the deterministic
cardinality threshold $R_{\min} := \log_2 K_{\min}$ induced by the
welfare-side quantization envelope (\Cref{prop:opposing}), while
$R_{\mathrm{Sh}}$ on the welfare-side source $\phi^*(T)$ gives the
corresponding Shannon-RDF floor or lower benchmark; upper endpoints
invoke $R_K$ at the
verification-driven cardinality, and the SK rate
$R_{\mathrm{SK}}^{\mathrm{cc}}$ in \Cref{cor:no-gap,cor:rate-gap} is
the SK rate~\eqref{eq:sk-rd} computed on the common-category encoder
class of \Cref{sec:mapping}, which equals $R_H$ under that
architecture and equals $R_K$ further under
\Cref{cor:rate-gap}'s uniform-assignment hypothesis. Because
$R_{\mathrm{SK}}^{\mathrm{cc}}$ is computed on this restricted
common-category architecture, it is an architectural specialisation of
the SK problem rather than the unrestricted SK rate; the corollaries'
claims attach to this restricted rate.
\Cref{tab:rates} summarises the four rate notions, their meanings,
equality conditions, and where they are used.

\begin{table}[t]
\caption{Rate notions used in this paper.}
\label{tab:rates}
\centering
\begin{tabular}{@{}p{0.14\linewidth}p{0.28\linewidth}p{0.22\linewidth}p{0.26\linewidth}@{}}
\toprule
Symbol & Meaning & Equality conditions & Where used \\
\midrule
$R_K := \log_2 K$ &
  partition-cardinality rate; bits to index a $K$-way partition &
  $R_K = R_H$ iff $C$ uniform &
  upper endpoint of band; \Cref{cor:rate-gap} under uniform
  assignment \\
$R_H := H(C)$ &
  entropy rate of the category indicator $C$ &
  $R_H = R_K$ iff $C$ uniform; $R_H \geq R_{\mathrm{Sh}}$ &
  SK rate on common-category architecture (see
  $R_{\mathrm{SK}}^{\mathrm{cc}}$ below) \\
$R_{\mathrm{Sh}}(D)$ &
  Shannon rate-distortion function; info-theoretic floor at
  distortion $D$ on the underlying source &
  $R_{\mathrm{Sh}}(D) = R_H$ when the induced category representation
  attains the Shannon RDF at distortion $D$ &
  Shannon-RDF floor/benchmark for the welfare-side endpoint;
  marginal lower bounds of \Cref{cor:no-gap}
  on $\phi^*(T)$ and on $T$ \\
$R_{\mathrm{SK}}^{\mathrm{cc}}$ &
  SK rate~\eqref{eq:sk-rd} restricted to the common-category
  encoder class of \Cref{sec:mapping} &
  $= R_H$ on this architecture; $= R_K$ further under uniform
  population assignment &
  \Cref{cor:no-gap,cor:rate-gap} \\
\bottomrule
\end{tabular}
\end{table}

\section{Three Reduction Propositions}\label{sec:reduction}

\Cref{prop:task-loss,prop:decoder,prop:encoder} record three places
where the SK apparatus reduces, under the designed-source mapping, to a
definition or a short proof.

All three propositions below, together with
\Cref{cor:no-gap,cor:rate-gap}, attach to the common-category
architecture of \Cref{sec:mapping}: deterministic finite-$K$ encoders
with $\hat{z} = c$ and $\hat{x} = \phi(c)$, $\phi$ injective. They are
architectural-specialisation results, not statements about the
unrestricted SK problem; in particular the rate
$R_{\mathrm{SK}}^{\mathrm{cc}}$ throughout is computed on this
restricted class.

\begin{proposition}[Task loss]\label{prop:task-loss}
Under A1--A2, for any category assignment $c$ with category allocation
$\phi(k) = \mathbb{E}[\phi^*(t) \mid c(t) = k]$,
\begin{equation}\label{eq:sandwich}
\frac{\alpha}{2}\,\varepsilon
\;\leq\; \mathbb{E}\bigl[U(t, \phi^*(t)) - U(t, \phi(c(t)))\bigr]
\;\leq\; \frac{\beta_U}{2}\,\varepsilon.
\end{equation}
\end{proposition}

\begin{proof}
Fix $t$ and let $k = c(t)$. Expand $U(t, \cdot)$ around $\phi^*(t)$:
\begin{align*}
U(t, \phi(k)) &= U(t, \phi^*(t))
  + \underbrace{\nabla_r U(t, \phi^*(t))^{\!\top}}_{= 0}
    \!\bigl(\phi(k) - \phi^*(t)\bigr) \\
  &\quad+ \tfrac{1}{2}\bigl(\phi(k) - \phi^*(t)\bigr)^{\!\top}
    \nabla^2_{rr} U(t, \tilde{r})\,
    \bigl(\phi(k) - \phi^*(t)\bigr)
\end{align*}
for some $\tilde{r}$ on $[\phi^*(t), \phi(k)]$. The linear term
vanishes by the first-order condition. By A1--A2,
$-\beta_U I \preceq \nabla^2_{rr} \preceq -\alpha I$, giving
\begin{align*}
\tfrac{\alpha}{2}\|\phi^*(t) - \phi(k)\|^2
&\leq U(t, \phi^*(t)) - U(t, \phi(k)) \\
&\leq \tfrac{\beta_U}{2}\|\phi^*(t) - \phi(k)\|^2.
\end{align*}
Take expectations over $t \sim F$.
\end{proof}

\begin{remark}[Status of \Cref{prop:task-loss}]\label{rem:sandwich-status}
The sandwich~\eqref{eq:sandwich} is a standard quadratic
upper-and-lower bound from convex analysis~\cite[Section~9.1]{boyd2004}
applied at the optimum of $U(t, \cdot)$. The novelty is not the lemma
itself but its role in the present setting. Once the semantic source is
the deterministic oracle allocation $\phi^*(t)$, the SK semantic
distortion $\mathbb{E}[d_s(x, \hat{x})]$ is controlled by the
within-category variance $\varepsilon$ via the sandwich, and equals
it up to a constant under quadratic utility; the upper-bound side
of~\eqref{eq:sandwich} is what enables the encoder reduction in
\Cref{prop:encoder}.
\end{remark}

\begin{proposition}[Decoder]\label{prop:decoder}
For any fixed deterministic $K$-partition
$c: \mathcal{T} \to \{1,\ldots,K\}$, the conditional mean
$\phi(k) = \mathbb{E}[\phi^*(t) \mid c(t) = k]$
minimises $\varepsilon$ among all decoders
$\hat{\phi}: \{1,\ldots,K\} \to \mathbb{R}^r$. For \emph{isotropic}
quadratic utility ($\alpha = \beta_U$, equivalently
$\nabla^2_{rr} U = -\alpha I$), the conditional mean minimises exact
welfare loss and the sandwich~\eqref{eq:sandwich} is tight. More
generally, for a weighted-quadratic utility
$U(t, r) = -\tfrac{1}{2}(r - \phi^*(t))^{\!\top} W (r - \phi^*(t))$
with constant positive-definite weight matrix $W$, the same
conditional mean also minimises the exact weighted-quadratic welfare
loss, although the Euclidean sandwich constants $\alpha, \beta_U$ need
not coincide. If the quadratic weight varies with $t$ or with the
cell, the exact welfare-minimising decoder becomes the corresponding
weighted conditional mean, while the ordinary conditional mean
continues to minimise the squared-oracle-error surrogate
of~\eqref{eq:sandwich}.
\end{proposition}

\begin{proof}
The conditional mean minimises
$\varepsilon = \mathbb{E}[\|\phi^*(t) - \hat{\phi}(c(t))\|^2]$ by the
MMSE property. Welfare loss is bounded above and below by $\varepsilon$
via the sandwich~\eqref{eq:sandwich}, so minimising $\varepsilon$
minimises both bounds simultaneously; for $\alpha = \beta_U$ the
sandwich constants coincide and the welfare-loss minimisation is
exact. For a constant positive-definite weight matrix $W$, the
per-cell first-order condition for the weighted-quadratic loss,
$W\,(\hat{\phi}(k) - \mathbb{E}[\phi^*(t) \mid c(t) = k]) = 0$, gives
$\hat{\phi}(k) = \mathbb{E}[\phi^*(t) \mid c(t) = k]$ independently of
$W$, extending the exact-welfare claim to the non-isotropic
constant-$W$ case even though the sandwich constants need not
coincide.
\end{proof}

\begin{proposition}[Encoder]\label{prop:encoder}
Assume $\phi^*(t)$ has finite second moment under $p_T$ and that the
source distribution puts no mass on cell boundaries. Any locally
$\varepsilon$-optimal $K$-partition of $\mathcal{T}$ with all cells of
positive probability satisfies the Lloyd-Max stationarity
conditions~\cite{lloyd1982,max1960} applied to $\phi^*(t)$:
nearest-neighbour cells in $\phi^*$-space and centroid means
$\phi(k) = \mathbb{E}[\phi^*(t) \mid c(t) = k]$. These conditions are
necessary, not sufficient.
\end{proposition}

\begin{proof}
The optimisation
$\min_{c} \mathbb{E}[\|\phi^*(t) - \phi(c(t))\|^2]$ over $K$-partitions
with conditional-mean centroids is the Lloyd-Max problem on $\phi^*(t)$
by definition; its first-order necessary conditions are the Lloyd-Max
stationarity conditions. In one dimension, additional conditions on
the source density are known under which Lloyd-Max stationary points
are globally optimal. For multi-dimensional sources such as the
smart-grid instance of \Cref{sec:smart-grid}, multiple local optima
are generic and the operational quantiser uses the iterative
generalisation of Linde, Buzo and Gray~\cite{lbg1980} with random
restarts to mitigate suboptimal fixed points. See the quantization
survey of Gray and Neuhoff~\cite{grayneuhoff1998} for the broader
theory.
\end{proof}

\paragraph{What \Cref{prop:task-loss,prop:decoder,prop:encoder} say
together.} Under the designed-source mapping, and on the deterministic
common-category encoder class of \Cref{sec:mapping}, SK's
exponential-tilting parametric decoder specialises to the conditional
mean on the squared-oracle-error surrogate (exact for isotropic
quadratic utility), the generalised Blahut-Arimoto encoder iteration
specialises to a Lloyd-Max iteration on $\phi^*(t)$, and the
Markov-chain-conditional identity
$R(D_s, D_o) = \max\{R(D_s), R(D_o)\}$ collapses to standard
rate-distortion on $\phi^*(t)$ for the single-distortion welfare
problem on this restricted class.

We do not claim the underlying tools (curvature sandwich, MMSE,
Lloyd-Max) as new contributions. The contribution of this section is
the identification that, on designed sources, the SK objects specialise
to the conditional-mean decoder and the Lloyd-Max stationarity
conditions on $\phi^*(t)$, recovering textbook quantities along
each axis.

\subsection{Algorithmic contrast}\label{sec:algo-contrast}

The reduction is also visible at the algorithm level. SK's framework is
computed by the generalised Blahut-Arimoto iteration of
\Cref{alg:sk-gba}, which alternates updates of the joint output marginal
$\nu(\hat{z}, \hat{x})$ driven by the coupled exponential weight
$A(z, \hat{z}, \hat{x})$ in~\eqref{eq:sk-gba}. The designed-source
encoder, by \Cref{prop:encoder}, is the Lloyd-Max iteration of
\Cref{alg:lloyd} on the oracle samples $\{\phi^*(t_i)\}_{i=1}^{M}$.
Throughout this paper $M$ denotes the population/sample size: the
number of agents in the worked examples of \Cref{sec:worked-example}
and \Cref{sec:smart-grid}, equivalently the Lloyd-Max training-sample
size in \Cref{alg:lloyd}, and the population partitioned by the
aggregate-testing setup of \Cref{def:agg-test}.

\begin{algorithm}[t]
\caption{SK generalised Blahut-Arimoto~\cite[Theorem~4]{sk2023}}
\label{alg:sk-gba}
\begin{algorithmic}[1]
\Require source $p(z)$; distortions $d_s, d_o$; multipliers
  $s_1, s_2 \leq 0$; tolerance $\tau > 0$
\State $\nu^{(0)}(\hat{z}, \hat{x}) \gets
  1/(|\hat{\mathcal{Z}}|\,|\hat{\mathcal{X}}|)$,\quad $k \gets 0$
\State $A(z, \hat{z}, \hat{x}) \gets
  \exp\!\bigl(s_1 d_s(z, \hat{x}) + s_2 d_o(z, \hat{z})\bigr)$
\Repeat
  \State update every entry of $\nu^{(k+1)}(\hat{z}, \hat{x})$
    via~\eqref{eq:sk-gba}
  \State $k \gets k + 1$
\Until{$\|\nu^{(k)} - \nu^{(k-1)}\|_1 \leq \tau$}
\State \Return $q^*(\hat{z}, \hat{x} \mid z)$ via~\eqref{eq:sk-tilting}
\end{algorithmic}
\end{algorithm}

\begin{algorithm}[t]
\caption{Lloyd-Max on $\phi^*$}\label{alg:lloyd}
\begin{algorithmic}[1]
\Require oracle samples $\{\phi^*(t_i)\}_{i=1}^{M}$; category count $K$
\State initialise centroids $\{\phi(k)\}_{k=1}^{K}$
\Repeat
  \State assign $c(t_i) \gets
    \arg\min_{k} \|\phi^*(t_i) - \phi(k)\|^2$ for $i = 1, \ldots, M$
  \State update $\phi(k) \gets |S_k|^{-1}
    \sum_{i \in S_k} \phi^*(t_i)$,\;
    $S_k = \{i : c(t_i) = k\}$
\Until{centroids stationary}
\State \Return centroids $\{\phi(k)\}$ and partition $\{S_k\}$
\end{algorithmic}
\end{algorithm}

\begin{table}[t]
\caption{Computational contrast at rate $R = \log_2 K$.}
\label{tab:complexity}
\centering
\begin{tabular}{lll}
\toprule
 & \Cref{alg:sk-gba} & \Cref{alg:lloyd} \\
\midrule
State &
  joint $\nu(\hat{z}, \hat{x})$ &
  $K$ centroids $\phi(k)$ \\
State size &
  $|\hat{\mathcal{Z}}|\,|\hat{\mathcal{X}}|$ &
  $K$ \\
Per-iter.\ cost &
  $O\!\bigl(|\mathcal{Z}|\,|\hat{\mathcal{Z}}|\,|\hat{\mathcal{X}}|\bigr)$ &
  $O(MK)$ \\
Couples $\hat{x}$, $\hat{z}$? &
  yes (via $A$) &
  no \\
Convergence &
  monotone in functional &
  monotone descent of $\varepsilon$ \\
Output &
  conditional law $q^*$ &
  centroids + partition \\
\bottomrule
\end{tabular}
\end{table}

The two iterations are summarised in \Cref{tab:complexity}.
\Cref{alg:sk-gba} maintains state on the product alphabet
$\hat{\mathcal{Z}} \times \hat{\mathcal{X}}$ because the exponential
weight $A$ couples the two distortions, and a per-iteration sweep costs
$O(|\mathcal{Z}|\,|\hat{\mathcal{Z}}|\,|\hat{\mathcal{X}}|)$.
\Cref{alg:lloyd} maintains only $K$ centroids in $\phi^*$-space and
costs $O(MK)$ per iteration because the semantic reconstruction
$\phi(c)$ is constrained, by \Cref{prop:decoder}, to be the conditional
mean of $\phi^*(t)$ within the category. The contrast is structural,
not merely a constant-factor speedup. The reduction removes the joint
exponential tilt and the dependence on $|\hat{\mathcal{X}}|$, replacing
both with a single MSE update on $K$ centroids.

\section{Opposing Monotonicity and the Feasibility Band}\label{sec:opposing}

\Cref{prop:task-loss,prop:decoder,prop:encoder} address the
single-distortion welfare problem. Designed mechanisms with a
verification side-constraint require both a welfare distortion and a
detection distortion. \Cref{prop:opposing} shows that the two move in
opposite directions in the rate variable $K$, so the joint constraint
is a feasibility band rather than a maximum.

\begin{definition}[Aggregate-testing setup]\label{def:agg-test}
An \emph{aggregate-testing setup} consists of: (a) a population of size
$M$ partitioned by the category map $c: \mathcal{T} \to \{1, \ldots, K\}$
with per-category pool size $m_k = |\{i : c(t_i) = k\}|$; (b) a
per-category aggregate statistic $T_k = m_k^{-1} \sum_{i \in
\mathrm{cat}_k} y_i$ used as test statistic; (c) per-agent measurement
variance $\sigma^2_{\mathrm{indiv}}$ and category-level shared-noise
variance $\sigma^2_{\mathrm{temp}}$, so that $T_k$ has variance
$\sigma^2_T(m_k) = \sigma^2_{\mathrm{temp}} + \sigma^2_{\mathrm{indiv}}/m_k$;
and (d) under the alternative $H_1$, a category-level mean shift
$\delta$. We write $\beta_{\mathrm{agg}}(K)$ for the detection power
of the one-sided $z$-test at level $\alpha_{\mathrm{FA}}$ under this
setup with uniform assignment ($m_k = M/K$).
\end{definition}

\begin{proposition}[Feasibility band under load-bearing partition]\label{prop:opposing}
Consider the designed-source mechanism class with deterministic
$K$-category encoders $c: \mathcal{T} \to \{1,\ldots,K\}$, paired
with the aggregate-testing side-constraint of \Cref{def:agg-test} on
the same partition $c$. Let
$\varepsilon^*(K) := \inf_{K\text{-partition}}
  \mathbb{E}\bigl[\|\phi^*(t) - \phi(c(t))\|^2\bigr]$
be the within-category-variance envelope over $K$-partitions.
$\varepsilon^*(K)$ is non-increasing in $K$. Under
uniform population assignment ($m_k = M/K$), the aggregate-test
detection power $\beta_{\mathrm{agg}}(K)$ is non-increasing in $K$.
The joint target pair $(\varepsilon^*, \beta^*)$ is therefore
achievable only if
\begin{equation}\label{eq:band}
K_{\min}(\varepsilon^*) \;\leq\; K \;\leq\; K_{\max}(\beta^*),
\end{equation}
with $K_{\min}(\varepsilon^*) := \inf\{K : \varepsilon^*(K) \leq
\varepsilon^*\}$ and $K_{\max}(\beta^*) := M/m^*(\beta^*)$ the
detection-feasibility ceiling at uniform assignment. The
uniform-assignment endpoints in~\eqref{eq:band} are the loosest such
envelope over population assignments (\Cref{rem:pool-vector}), and non-uniform
assignments weakly contract the band, strictly so when the minimum
pool size falls below the uniform value and the detection-side
constraint is binding. The joint design problem is
therefore not a standard two-distortion rate-distortion problem,
because the cardinality $K$ enters as an external constraint on the
admissible-encoder class via the load-bearing-partition assumption.
\end{proposition}

\begin{proof}
(i) Welfare side. Any $K$-partition lifts to a $(K+1)$-partition by
splitting one cell, with the refined partition free to choose its
category allocation $\phi$ independently on each sub-cell. Applied to
the within-category variance, the refinement weakly reduces the sum
of within-cell variances, so $\varepsilon^*(K)$ is non-increasing in
$K$. A welfare target $\varepsilon^*(K) \leq \varepsilon^*$ therefore
requires $K \geq K_{\min}(\varepsilon^*)$.
(ii) Verification side. Under \Cref{def:agg-test} with uniform
population assignment, $m_k = M/K$ decreases in $K$, so
$\sigma_T^2(m_k)$ increases in $K$, and the detection power
$\beta_{\mathrm{agg}}(K)$ decreases in $K$. A detection target
$\beta_{\mathrm{agg}}(K) \geq \beta^*$ therefore requires
$K \leq K_{\max}(\beta^*)$.
(iii) Combining, $K_{\min} \leq K \leq K_{\max}$ is necessary for the
joint target pair to be achievable under uniform assignment. The
band is non-empty iff $K_{\min} \leq K_{\max}$.
\end{proof}

\begin{remark}[Lower endpoint reads on the surrogate envelope]\label{rem:surrogate-endpoint}
The lower endpoint $K_{\min}(\varepsilon^*)$ in \eqref{eq:band} reads
the welfare target $\varepsilon^*$ on the within-category-variance
envelope $\varepsilon^*(K)$, a surrogate for exact welfare loss. Under
the conditional-mean decoder of \Cref{prop:decoder}, exact welfare loss
is controlled sandwich-wise by $\varepsilon^*$ via
\Cref{prop:task-loss} but need not be monotone in $K$ at non-quadratic
utility; under quadratic utility the conditional-mean decoder is exact
and the band reads directly on welfare loss. The band is therefore a
cardinality-level necessary condition stated on the surrogate, exact on
welfare under quadratic utility and otherwise sandwich-controlled.
\end{remark}

\begin{remark}[Non-uniform assignment: pool-size vector form]\label{rem:pool-vector}
\Cref{prop:opposing} restricts to uniform population assignment. For a
general assignment with pool-size vector $(m_1, \ldots, m_K)$ summing
to $M$, the \Cref{def:agg-test} detection condition binds at the
weakest pool $m_{\min} := \min_k m_k$, since $\sigma_T^2(m_k) =
\sigma_{\mathrm{temp}}^2 + \sigma_{\mathrm{indiv}}^2 / m_k$ is
decreasing in $m_k$ and the aggregate test must meet the target power
$\beta^*$ at every category. Writing $m^*(\beta^*)$ for the threshold
pool size at which the detection target is exactly met, the
verification-side constraint is $m_{\min} \geq m^*(\beta^*)$. At fixed
$K$, $m_{\min}$ is maximised at uniform assignment with value $M/K$, so
the uniform-assignment endpoint $K_{\max}(\beta^*) = M / m^*(\beta^*)$
is the loosest detection-side upper bound on $K$ available across
population assignments. Two partitions sharing the same $K$ or the same
entropy $H(C)$ but differing in their pool-size vectors carry different
verification-side cuts: pool-size heterogeneity strictly contracts the
band. The uniform-assignment band of \Cref{prop:opposing} is therefore
a cardinality-level outer envelope, attained only when the realised
partition has balanced pool sizes, or when an external balancing
mechanism enforces them. A generic welfare-side Lloyd-Max partition
need not satisfy this condition, and pool-size heterogeneity contracts
the operational detection endpoint.
\end{remark}

\subsection{Worked example: feasibility band on a Gaussian designed source}\label{sec:worked-example}

We illustrate \Cref{prop:opposing} on a Gaussian designed source. The
demand type is $t \sim \mathcal{N}(0, \sigma_t^2)$ and the oracle
allocation is $\phi^*(t) = t$ after centring, so
$\phi^*(T) \sim \mathcal{N}(0, \sigma_\phi^2)$ with
$\sigma_\phi^2 = \sigma_t^2$. The population contains $M$ agents,
divided into $K$ categories by the Lloyd-Max quantiser of
\Cref{prop:encoder} applied to $\phi^*(t)$. Each category contains
$m_k = M/K$ agents. By \Cref{rem:pool-vector}, this uniform assignment
is the balanced-category outer envelope of the feasibility band, not a
separate design choice layered on Lloyd-Max. A realised Lloyd-Max
partition may give nonuniform category masses and therefore a stricter
verification ceiling than the band reported below.

\paragraph{Welfare side: $R_{\min}^{\mathrm{Sh}}(\varepsilon^*)$.}
For a Gaussian source, the Shannon rate-distortion
function~\cite[Theorem~10.3.2]{coverthomas2006} is
\begin{equation}\label{eq:rdf}
R(D) \;=\; \tfrac{1}{2} \log_2 \bigl(\sigma_\phi^2 / D\bigr),
\end{equation}
giving $D(R) = \sigma_\phi^2 \cdot 2^{-2R}$. This is the scalar
Gaussian Shannon rate-distortion envelope on $\phi^*(T)$, not the
finite-$K$ Lloyd-Max distortion. Identifying $\varepsilon$ with $D$,
the welfare-target constraint $\varepsilon \leq \varepsilon^*$
is satisfied at any rate
\begin{equation}\label{eq:R-min}
R \;\geq\; R_{\min}^{\mathrm{Sh}}(\varepsilon^*) \;:=\;
R_{\mathrm{Sh}}(\varepsilon^*) \;=\;
  \tfrac{1}{2} \log_2 \bigl(\sigma_\phi^2 / \varepsilon^*\bigr),
\end{equation}
the Shannon rate-distortion floor on the welfare-side source. Under
uniform assignment of $C$, the partition-cardinality rate and the
entropy rate coincide ($R_H = R_K = \log_2 K$), and this Shannon-derived
rate constraint instantiates as the pool-size constraint
\begin{equation}\label{eq:K-min}
K \;\geq\; K_{\min}^{\mathrm{Sh}}(\varepsilon^*) \;:=\; 2^{R_{\min}^{\mathrm{Sh}}(\varepsilon^*)}
  \;=\; \sigma_\phi / \sqrt{\varepsilon^*}.
\end{equation}
$R_{\min}^{\mathrm{Sh}}$ and $K_{\min}^{\mathrm{Sh}}$ are decreasing in
$\varepsilon^*$, as expected. Tighter welfare targets demand more
rate, equivalently more categories. The operational
$K_{\min}(\varepsilon^*)$ of \Cref{prop:opposing} is obtained from the
empirical Lloyd-Max within-category-variance envelope rather than from
the Shannon benchmark, and Equation~\eqref{eq:R-min} is therefore a
Shannon lower benchmark on the operative cardinality threshold
$R_{\min}(\varepsilon^*) := \log_2 K_{\min}(\varepsilon^*)$ of
\Cref{cor:rate-gap}, with equality only when the finite-$K$ encoder
is Shannon-RDF-optimal.

\paragraph{Detection side: $K_{\max}(\beta^*)$.}
The verification side-constraint requires that an aggregate test on
each category detect a category-level effect of size $\delta$ with
power at least $\beta^*$ at false-alarm level $\alpha_{\mathrm{FA}}$
(disambiguated from the strong-concavity constant $\alpha$ of A2).
Within
category $k$, the aggregate statistic
\[
T_k \;=\; \frac{1}{m_k} \sum_{i \in \mathrm{cat}_k} y_i
\]
has variance
$\sigma_T^2(K) = \sigma_{\mathrm{temp}}^2 + \sigma_{\mathrm{indiv}}^2 K / M$,
where $\sigma_{\mathrm{temp}}^2$ is a category-level shared-noise
component and $\sigma_{\mathrm{indiv}}^2$ is per-agent measurement
noise. The detection power at threshold
$z_{1-\alpha_{\mathrm{FA}}}\,\sigma_T(K)$ under a category-level mean shift $\delta$
is
\begin{equation}\label{eq:beta-K}
\beta_{\mathrm{agg}}(K)
  \;=\; \Phi\!\left(\frac{\delta}{\sigma_T(K)} - z_{1-\alpha_{\mathrm{FA}}}\right).
\end{equation}
The constraint $\beta_{\mathrm{agg}} \geq \beta^*$ is satisfied when
\begin{equation}\label{eq:K-max}
K \;\leq\; K_{\max}(\beta^*) \;:=\;
  M \cdot
  \frac{\delta^2 / (z_{1-\alpha_{\mathrm{FA}}} + z_{\beta^*})^2
        - \sigma_{\mathrm{temp}}^2}{\sigma_{\mathrm{indiv}}^2},
\end{equation}
provided the bracketed numerator is positive. $K_{\max}$ is decreasing
in $\beta^*$. Tighter power targets demand larger pools and therefore
fewer categories. If
$\delta^2 / (z_{1-\alpha_{\mathrm{FA}}} + z_{\beta^*})^2 \leq \sigma_{\mathrm{temp}}^2$
the test fails even at $K = 1$, and the verification side-constraint
is infeasible irrespective of $\varepsilon^*$. The right-hand side
of~\eqref{eq:K-max} is real-valued; the achievable integer rate is
$K \leq \lfloor K_{\max}(\beta^*) \rfloor$, and the welfare-side
counterpart in~\eqref{eq:K-min} is $K \geq \lceil K_{\min}^{\mathrm{Sh}}(\varepsilon^*)
\rceil$. We retain the real-valued expressions where they simplify the
monotonicity arguments, with the floor/ceiling understood at the
integer-$K$ instantiation.

\paragraph{Feasibility band.}
Combining \eqref{eq:R-min} and \eqref{eq:K-max}, within the idealised
scalar Gaussian RDF model and uniform assignment, the target pair
$(\varepsilon^*, \beta^*)$ lies in the Shannon-benchmark feasibility
envelope iff
\begin{equation}\label{eq:band-condition}
R_{\min}^{\mathrm{Sh}}(\varepsilon^*) \;\leq\; \log_2 K_{\max}(\beta^*),
\end{equation}
equivalently in pool-size form,
\begin{equation}\label{eq:band-condition-K}
\frac{\sigma_\phi}{\sqrt{\varepsilon^*}}
  \;\leq\;
  M \cdot
  \frac{\delta^2 / (z_{1-\alpha_{\mathrm{FA}}} + z_{\beta^*})^2
        - \sigma_{\mathrm{temp}}^2}{\sigma_{\mathrm{indiv}}^2}.
\end{equation}
The band is non-empty for slack target pairs and shrinks as either
target tightens. The critical welfare target at which the band
collapses, with $\beta^*$ fixed, is
\begin{equation}\label{eq:eps-critical}
\varepsilon^*_{\mathrm{c}}(\beta^*)
  \;=\; \sigma_\phi^2 \,/\, K_{\max}(\beta^*)^2.
\end{equation}

\paragraph{Numerical instance.}
\Cref{fig:band} shows two regimes for the parameter setting
$\sigma_\phi^2 = 1$, $M = 200$, $\sigma_{\mathrm{temp}}^2 = 0.02$,
$\sigma_{\mathrm{indiv}}^2 = 1$, $\delta = 1$, $\alpha_{\mathrm{FA}} = 0.05$,
$\beta^* = 0.8$ (so $z_{1-\alpha_{\mathrm{FA}}} + z_{\beta^*} \approx 2.49$,
$K_{\max} \approx 28.3$, and
$\varepsilon^*_{\mathrm{c}} \approx 1.24 \times 10^{-3}$). Panel~(a)
sets $\varepsilon^* = 5 \times 10^{-3}$, giving the Shannon-benchmark
band $K \in [14.1,\, 28.3]$ shaded green. The operational lower
endpoint for a finite-$K$ deterministic quantiser would be read from
the empirical Lloyd-Max envelope. Panel~(b) sets
$\varepsilon^* = 8 \times 10^{-4}$, below the benchmark critical value,
so $K_{\min}^{\mathrm{Sh}} \approx 35.4 > K_{\max} \approx 28.3$. Since
the operational deterministic lower endpoint cannot lie below the
Shannon floor, the joint target is infeasible.

\begin{figure}[t]
\centering
\includegraphics[width=\linewidth]{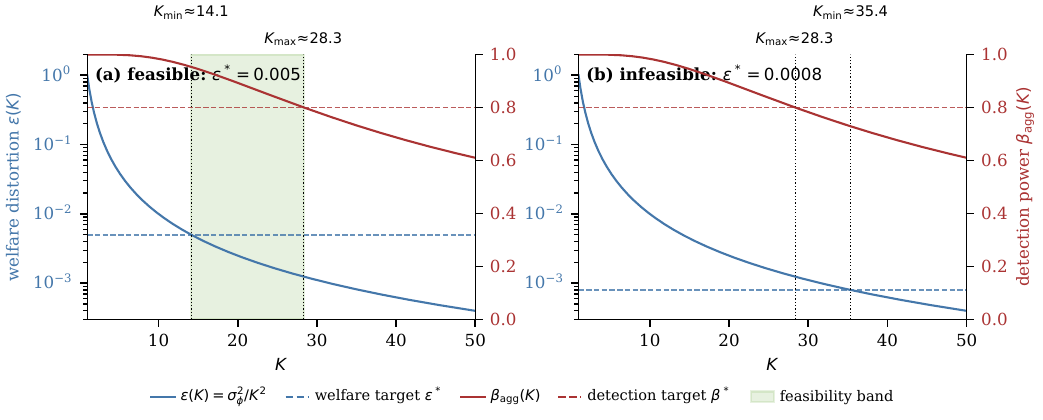}
\caption{Shannon-benchmark feasibility band
$[K_{\min}^{\mathrm{Sh}}(\varepsilon^*), K_{\max}(\beta^*)]$ for the
Gaussian designed-source instance. The Gaussian Shannon RDF envelope
$D(\log_2 K) = \sigma_\phi^2 / K^2$ on $\phi^*(T)$ (left log-axis,
blue) and the aggregate detection power $\beta_{\mathrm{agg}}(K)$
(right linear-axis, red) move in opposite directions in $K$.
Horizontal dashed lines mark $\varepsilon^*$ and $\beta^* = 0.8$, and
vertical dotted lines locate $K_{\min}^{\mathrm{Sh}}$ and $K_{\max}$.
(a) Non-empty band at $\varepsilon^* = 5 \times 10^{-3}$.
(b) Collapsed band at $\varepsilon^* = 8 \times 10^{-4}$, where
$K_{\min}^{\mathrm{Sh}}$ exceeds $K_{\max}$.}
\label{fig:band}
\end{figure}

\subsection{Rate equivalence under designed-source mapping}\label{sec:rate-equivalence}

\Cref{prop:opposing} characterises the cardinality-level outer band for
designed-source dual-purpose signals in the opposing-monotonicity
regime, where the verification-side mapping $R \mapsto D_o(R)$ is not
monotone non-increasing. The complementary question is whether
the designed-source reduction differs from SK's
formula~\eqref{eq:sk-rd} in the standard regime, where the
verification-side mapping is monotone non-increasing and $R(T; D_o)$ is
a standard rate-distortion function. The following lemma and corollary
show that the two characterisations agree on the Shannon-RDF lower-bound
structure inside this regime, with exact agreement under compatible
RDF-optimal common-category reconstructions, and identify opposing
monotonicity as a structurally distinct regime
in which SK's rate-monotonicity scope is not met and a complementary
characterisation is required. Other structural exits from the SK
admissible class are also possible, for instance non-single-letter,
finite-block, or adversarial distortion structures; opposing
monotonicity is the regime considered here.

\begin{lemma}[Indirect rate-distortion reformulation]\label{lem:indirect}
Assume $\mathcal{T}$ is a Polish space, $\phi^*: \mathcal{T} \to
\mathbb{R}^r$ is Borel, and a regular conditional distribution for $t$
given $\phi^*(t)$ exists (automatic under these Polish-space
hypotheses). Let
$\ell: \mathcal{T} \times \mathbb{R}^r \to \mathbb{R}_{\geq 0}$ be
jointly Borel-measurable with $\mathbb{E}[\ell(t, \hat{x})] < \infty$
for every measurable $\hat{x}$, and define
\begin{equation}\label{eq:ds-star}
d_s^*(x, \hat{x}) \;:=\;
  \mathbb{E}\bigl[\,\ell(t, \hat{x}) \,\big|\, \phi^*(t) = x\,\bigr].
\end{equation}
Then $d_s^*$ is non-negative and measurable on
$\mathcal{X} \times \mathbb{R}^r$, hence an admissible SK semantic
distortion in~\eqref{eq:sk-rd}. Suppose the encoder
$c: \mathcal{T} \to \{1, \ldots, K\}$ factors through $\phi^*$, in the
sense that $c(t) = \tilde{c}(\phi^*(t))$ for some measurable
$\tilde{c}: \mathbb{R}^r \to \{1, \ldots, K\}$ (equivalently, $c$ is
$\sigma(\phi^*)$-measurable). Then for any decoder
$g: \{1, \ldots, K\} \to \mathbb{R}^r$, the reconstruction
$\hat{X} := g(c(t))$ satisfies
\begin{equation}\label{eq:tower}
\mathbb{E}\bigl[\,d_s^*(\phi^*(t),\, \hat{X})\,\bigr]
\;=\; \mathbb{E}\bigl[\,\ell(t,\, \hat{X})\,\bigr].
\end{equation}
\end{lemma}

\begin{proof}
Measurability and non-negativity of $d_s^*$ follow from those of
$\ell$ by standard properties of conditional expectation. Under the
factorisation $c(t) = \tilde{c}(\phi^*(t))$, the reconstruction
$\hat{X} = g(\tilde{c}(\phi^*(t)))$ is
$\sigma(\phi^*(t))$-measurable. Pulling $\hat{X}$ out of the
conditional expectation yields
$\mathbb{E}[\ell(t, \hat{X}) \mid \phi^*(t)]
 = \mathbb{E}[\ell(t, x) \mid \phi^*(t)]\big|_{x = \hat{X}}
 = d_s^*(\phi^*(t), \hat{X})$.
The tower property then gives
$\mathbb{E}[\ell(t, \hat{X})]
 = \mathbb{E}\bigl[\mathbb{E}[\ell(t, \hat{X}) \mid \phi^*(t)]\bigr]
 = \mathbb{E}[d_s^*(\phi^*(t), \hat{X})]$.
\end{proof}

\begin{remark}[Lloyd-Max satisfies the factorisation]\label{rem:lloyd-factor}
The Lloyd-Max encoder of \Cref{prop:encoder} is by construction a
quantiser of $\phi^*(t)$, so its category function factors through
$\phi^*$ in the sense required by \Cref{lem:indirect}. Thus
\eqref{eq:tower} applies to the deterministic common-category encoders
considered in this paper. The lemma identifies the welfare loss with
the induced distortion on $\phi^*(T)$ for such encoders, and does not
assert unrestricted Shannon-RDF optimality of a finite-$K$ Lloyd-Max
partition.
\end{remark}

\Cref{lem:indirect} is the indirect rate-distortion reformulation of
Witsenhausen~\cite{witsenhausen1980}, specialised to the deterministic
semantic mapping $x = \phi^*(t)$. The same construction underlies the
indirect characterisation of Liu, Shao, Zhang and
Poor~\cite{lszp2022} reviewed in \Cref{sec:lzp}.

\begin{corollary}[SK specialisation under compatible designed-source mapping]\label{cor:no-gap}
Consider the common-category architecture of \Cref{sec:mapping}, in
which both reconstructions are deterministic functions of a shared
category index $c$ ($\hat{z} = c$, $\hat{x} = \phi(c)$ with $\phi$
injective). Write $R_{\mathrm{SK}}^{\mathrm{cc}}(D_s, D_o)$ for the SK
rate~\eqref{eq:sk-rd} computed over this common-category encoder
class, not over the unconstrained SK admissible class. For any welfare
loss $\ell$ satisfying the conditions of \Cref{lem:indirect}, any
observation distortion $d_o$ on $T$ for which $R(T; D_o)$ is a
standard rate-distortion function, and setting $d_s = d_s^*$
in~\eqref{eq:ds-star},
\begin{equation}\label{eq:no-gap}
R_{\mathrm{SK}}^{\mathrm{cc}}(D_s,\, D_o)
\;\geq\; \max\!\bigl\{\,
        R\bigl(\phi^*(T);\, D_s\bigr),\;
        R\bigl(T;\, D_o\bigr)
      \bigr\},
\end{equation}
where $R(\phi^*(T); D_s)$ is the Shannon rate-distortion function on
$\phi^*(t)$ at semantic-distortion target $D_s$ and $R(T; D_o)$ is the
Shannon rate-distortion function on $T$ at observation-distortion
target $D_o$. Equality holds when the semantic-side and
observation-side marginal optimisers are \emph{compatible}, i.e.\ when
a single common-category encoder simultaneously attains both marginal
rate-distortion targets at the rate of the larger marginal. In the
squared-oracle-error specialisation
$d_o(t, c) := \|\phi^*(t) - \phi(c)\|^2$ (written as $d_o = d_s^*$
below), both marginal deterministic common-category problems coincide
on $\phi^*(T)$, so a common Lloyd-Max quantiser that solves either
restricted-deterministic problem at the larger target solves both.
Equality with the unrestricted Shannon rate-distortion functions
$R(\phi^*(T); D_s)$, $R(T; D_o)$ then requires the additional
condition that this finite-$K$ common-category encoder be
Shannon-RDF-optimal, which is not implied by Lloyd-Max stationarity
alone.
\end{corollary}

\begin{proof}
By the variable mapping of \Cref{sec:mapping}, $\hat{x} = \phi(c)$ and
$\hat{z} = c$ are deterministic functions of $c = c(t)$, and with
$\phi$ injective each is a deterministic function of the other. Both
SK Markov chains~\eqref{eq:sk-markov} therefore hold automatically on
the common-category architecture.

\emph{Lower bound.} Define the marginal common-category rates
$R_{\mathrm{SK}}^{\mathrm{cc}}(D_s) := \min_{c} \{ I(z; \hat{x}) :
  \mathbb{E}[d_s^*(\phi^*(z), \hat{x})] \leq D_s,\ \hat{x} = \phi(c)
  \text{ common-category}\}$
and $R_{\mathrm{SK}}^{\mathrm{cc}}(D_o)$ analogously over the
single-distortion observation-side problem on the same encoder class.
Adding the other distortion constraint to either marginal can only
raise the minimum, so $R_{\mathrm{SK}}^{\mathrm{cc}}(D_s, D_o) \geq
\max\{R_{\mathrm{SK}}^{\mathrm{cc}}(D_s),
R_{\mathrm{SK}}^{\mathrm{cc}}(D_o)\}$. The observation-side marginal therefore satisfies
$R_{\mathrm{SK}}^{\mathrm{cc}}(D_o) \geq R(T; D_o)$, the unrestricted
Shannon rate-distortion function on $T$ at observation-distortion
target $D_o$. The common-category architecture restricts the encoder
class to deterministic finite-$K$ assignments, and class restriction
can only raise the minimum. Equality holds under the compatibility
hypothesis treated below. The semantic-side
marginal lower bound follows in two steps. (i) The unrestricted
semantic marginal RDF on $z$ at target $D_s$ equals the Shannon RDF on
$\phi^*$. Since $d_s^*(\phi^*(z), \hat{x})$ depends on $z$ only through
$\phi^*(z)$, the random variable $\phi^*(z)$ is a sufficient statistic
for $z$ with respect to the semantic distortion. For any stochastic
encoder $q(\hat{x} \mid z)$ feasible at target $D_s$, the averaged
kernel
$\tilde{q}(\hat{x} \mid \phi^*(z))
  := \mathbb{E}_{z \mid \phi^*(z)}\!\bigl[\, q(\hat{x} \mid z)\,\bigr]$
preserves the $(\phi^*(z), \hat{x})$ marginal of the induced joint, and
therefore the expected semantic distortion, while satisfying
$I_{\tilde{q}}(z; \hat{x}) = I(\phi^*(z); \hat{x})
  \leq I_{q}(z; \hat{x})$
by the data-processing inequality applied along
$\hat{x} \to z \to \phi^*(z)$. The unrestricted semantic marginal RDF
therefore equals $R(\phi^*(T); D_s)$, the Shannon rate-distortion
function on $\phi^*(T)$ at semantic-distortion target $D_s$.
(ii) The common-category deterministic class of \Cref{sec:mapping} is
a restriction of this unrestricted class, so class restriction gives
$R_{\mathrm{SK}}^{\mathrm{cc}}(D_s) \geq R(\phi^*(T); D_s)$, with
equality requiring that some deterministic common-category encoder be
Shannon-RDF-optimal at $D_s$.
\Cref{lem:indirect} then identifies
$\mathbb{E}[d_s^*(\phi^*(t), \hat{x})]$ with the expected welfare loss
$\mathbb{E}[\ell(t, \hat{x})]$, so the semantic-side target can be
expressed equivalently as a welfare-loss target or as a distortion
target on $\phi^*$-space.

\emph{Equality under compatibility.} Suppose a single common-category
encoder $c^\dagger$ realises both marginal targets simultaneously,
i.e.\ both $\mathbb{E}[d_s^*(\phi^*(z), \phi(c^\dagger))] \leq D_s$ and
$\mathbb{E}[d_o(z, c^\dagger)] \leq D_o$ at rate
$I(z; c^\dagger) = \max\{R(\phi^*(T); D_s),\, R(T; D_o)\}$. Then
$c^\dagger$ is feasible for the joint minimisation, so
$R_{\mathrm{SK}}^{\mathrm{cc}}(D_s, D_o) \leq I(z; c^\dagger) =
\max\{R(\phi^*(T); D_s),\, R(T; D_o)\}$, giving equality with the
lower bound. In the squared-oracle-error case $d_o = d_s^*$, both
marginal deterministic common-category problems coincide on
$\phi^*(T)$, so a common Lloyd-Max quantiser at the larger
restricted-class target attains both marginal distortions
simultaneously, establishing the restricted-class compatibility
hypothesis. Equality with the unrestricted Shannon rate-distortion
functions in the bound~\eqref{eq:no-gap} additionally requires that
this finite-$K$ encoder be Shannon-RDF-optimal.
\end{proof}

\begin{corollary}[Feasibility band under opposing monotonicity]\label{cor:rate-gap}
Under the opposing-monotonicity hypotheses of \Cref{prop:opposing},
with uniform category masses $\Pr(C = k) = 1/K$ (attained exactly by
balanced assignment and up to sampling fluctuations under empirical
$m_k = M/K$), and the common-category architecture $\hat{z} = c$,
$\hat{x} = \phi(c)$ of \Cref{cor:no-gap}, the SK rate specialises to
the partition-cardinality rate
$R_{\mathrm{SK}}^{\mathrm{cc}} = R_H = R_K = \log_2 K$, and the target
pair $(\varepsilon^*, \beta^*)$ is achievable at rate $R$ only if
\begin{equation}\label{eq:rate-band}
R_{\min}(\varepsilon^*) \;\leq\; R \;\leq\; R_{\max}(\beta^*),
\qquad R_{\max}(\beta^*) := \log_2 K_{\max}(\beta^*).
\end{equation}
The lower endpoint is the deterministic cardinality-rate threshold
$R_{\min}(\varepsilon^*) := \log_2 K_{\min}(\varepsilon^*)$ from
\Cref{prop:opposing}, where $K_{\min}(\varepsilon^*)$ is the smallest
$K$ at which the Lloyd-Max within-category-variance envelope
$\varepsilon^*(K)$ meets the welfare target. The Shannon
rate-distortion floor on $\phi^*(T)$ at the same target is a lower
benchmark, $R_{\mathrm{Sh}}(\varepsilon^*) \leq R_{\min}(\varepsilon^*)$,
with equality only when the finite-$K$ common-category encoder is
Shannon-RDF-optimal. The upper
endpoint is the partition-cardinality rate at the verification-driven
uniform-assignment maximum, $\log_2 K_{\max}(\beta^*)$. The two
endpoints are therefore different operational rate notions; the band
width $R_{\max} - R_{\min} = \log_2\bigl(K_{\max}/K_{\min}\bigr)$ bits
is the gap between the deterministic welfare-side cardinality endpoint
and the verification-driven cardinality ceiling, with the band
non-empty when this gap is non-negative and empty when
$R_{\min} > R_{\max}$. The band is operational on a realised
partition only when the welfare-side optimal $K$-partition also
satisfies the pool-size-vector condition of \Cref{rem:pool-vector};
uniform population assignment is the loosest such case and is the
assignment used in the scalar Gaussian worked example of
\Cref{sec:worked-example} and the smart-grid instance of
\Cref{sec:smart-grid}. Lloyd-Max optimality on the welfare side does
not by itself certify uniform assignment, so feasibility on the lower
boundary requires a separate population-assignment hypothesis.

The verification-side quantity $D_o := 1 - \beta_{\mathrm{agg}}$ is a
function of the aggregate test statistic over the pool of size $m_k$
in each category, not of a per-letter $(z, \hat{z})$ pair, so $D_o$
is not a single-letter distortion of the form admitted by SK. The
joint problem leaves the SK admissible class on this structural
ground. The rate-monotonicity violation under \Cref{prop:opposing},
namely that the mapping $R \mapsto 1 - \beta_{\mathrm{agg}}(R)$ is
non-decreasing whereas a standard rate-distortion function $R(D)$ is
non-increasing in $D$, is the downstream symptom. The
constrained-design band of the present corollary is the appropriate
characterisation in this regime.
\end{corollary}

\begin{proof}
Under the common-category architecture $\hat{z} = c$, $\hat{x} = \phi(c)$
with $\phi$ injective, the SK rate is
$R_{\mathrm{SK}}^{\mathrm{cc}} = I(z; \hat{z}, \hat{x}) = I(t; c) = H(C)$,
the last equality because $c$ is a deterministic function of $t$.
Under uniform population assignment, $H(C) = \log_2 K$, so
$R_{\mathrm{SK}}^{\mathrm{cc}} = R_H = R_K = \log_2 K$. By
\Cref{prop:opposing},
$K_{\min}(\varepsilon^*) \leq K \leq K_{\max}(\beta^*)$ is necessary
for joint achievability under uniform assignment, and taking $\log_2$
gives the necessary-condition band $R_{\min} \leq R \leq R_{\max}$
with $R_{\min}(\varepsilon^*) := \log_2 K_{\min}(\varepsilon^*)$ and
$R_{\max}(\beta^*) := \log_2 K_{\max}(\beta^*)$, both deterministic
cardinality endpoints of \Cref{prop:opposing}. The Shannon-RDF floor on
$\phi^*(T)$ at the same target satisfies
$R_{\mathrm{Sh}}(\varepsilon^*) \leq R_{\min}(\varepsilon^*)$, with
equality only when the finite-$K$ encoder is Shannon-RDF-optimal. The
width and emptiness characterisations restate parts of
\Cref{prop:opposing}.

The single-letter argument. By \Cref{def:agg-test}, $\beta_{\mathrm{agg}}$
is a function of the aggregate test statistic
$T_k = m_k^{-1} \sum_{i \in \mathrm{cat}_k} y_i$ over the entire pool
of size $m_k$, so $D_o := 1 - \beta_{\mathrm{agg}}$ has no
representation as a per-letter distortion measure on
$\mathcal{Z} \times \hat{\mathcal{Z}}$; the SK framework requires
$d_o(z, \hat{z})$ to be of that single-letter form, so $D_o$ falls
outside the SK admissible class. The rate-monotonicity violation is
the downstream symptom: by \Cref{prop:opposing} $\beta_{\mathrm{agg}}(K)$
decreases in $K$, so $D_o$ increases in $R$, whereas a standard
rate-distortion function $R(D)$ is non-increasing in $D$ (equivalently,
its inverse $D(R)$ is non-increasing in $R$). Hence
$R \mapsto D_o(R)$ does not satisfy the monotonicity required of an SK
marginal rate-distortion function.
\end{proof}

\paragraph{Numerical instances of the band.}
The values below report Shannon-/Gaussian-benchmark realisations of the
welfare-side endpoint of \Cref{cor:rate-gap}, evaluated using the
scalar Gaussian Shannon RDF of \Cref{sec:worked-example} and the
vector Gaussian upper bound of \Cref{sec:smart-grid}; these equal the
operational $R_{\min}$ of \Cref{cor:rate-gap} only when the finite-$K$
encoder is Shannon-RDF-optimal.
For the scalar Gaussian designed source ($d = 1$) of
\Cref{sec:worked-example} Panel~(a) in \Cref{fig:band}
($\varepsilon^* = 5 \times 10^{-3}$, $\beta^* = 0.8$),
$R_{\min}^{\mathrm{Sh}} \approx 3.82$~bits and $R_{\max} \approx 4.82$~bits,
giving a band of width
$R_{\max} - R_{\min}^{\mathrm{Sh}} \approx 1.00$~bit; equivalently
$K_{\min}^{\mathrm{Sh}} \approx 14.1$ and $K_{\max} \approx 28.3$. The
vector smart-grid instance of \Cref{sec:smart-grid} at $d = 4$ with
empirical per-coordinate $\sigma_\phi = 0.390$, applying the vector
Gaussian rate-distortion upper bound
$R_{\min}^{\mathrm{G}}(\varepsilon^*) =
\tfrac{d}{2}\log_2(d\,\sigma_\phi^2/\varepsilon^*)$ (sufficient at
$\varepsilon^*$; see \Cref{sec:smart-grid}), has
$R_{\min}^{\mathrm{G}} \approx 3.21$~bits and $R_{\max} \approx 4.82$~bits
at the same verification parameters with $\varepsilon^* = 0.20$, band
width $\approx 1.62$~bits
($K_{\min}^{\mathrm{G}} \approx 9.24$, $K_{\max} \approx 28.3$).
The SK formula, applied mechanically with $R(D_o)$ interpreted as the
smallest rate achieving the verification target, collapses to the
single value $R_{\min}^{\mathrm{Sh}}$ on the lower edge of the band: at $K = 1$
the aggregate-test power exceeds $\beta^*$ at the parameter settings of
\Cref{fig:band}, so $R(D_o) = 0$ and the welfare side dominates the
maximum. The band's full extent, and the additional bit of design
freedom on the upper side, are not surfaced. For Panel~(b) of
\Cref{fig:band} ($\varepsilon^* = 8 \times 10^{-4}$, $\sigma_\phi = 1$),
$R_{\min}^{\mathrm{Sh}} \approx 5.14$~bits exceeds $R_{\max} \approx 4.82$~bits
and the band is empty
($K_{\min}^{\mathrm{Sh}} \approx 35.4 > K_{\max} \approx 28.3$).
The SK formula returns $R_{\min}^{\mathrm{Sh}} \approx 5.14$~bits, at which
rate $\beta_{\mathrm{agg}} \approx 0.73 < \beta^*$ and the verification
target is violated; the joint infeasibility of the target pair is not
visible from $R_{\mathrm{SK}}^{\mathrm{cc}}$ alone.

\paragraph{Uniqueness of the divergence.}
\Cref{cor:no-gap,cor:rate-gap} together with \Cref{prop:opposing} pin
down the structural relationship between the SK characterisation and
the designed-source reduction along the monotonicity axis. When both
single-distortion rate-distortion functions are well-defined,
specifically when the verification-side mapping
$R \mapsto D_o(R)$ is monotone decreasing as a standard rate-distortion
function requires, \Cref{cor:no-gap} lower-bounds the common-category
SK rate by the max of the two corresponding Shannon rate-distortion
functions on $\phi^*(T)$ and on $T$, with equality under compatible
common-category reconstructions; in the squared-oracle-error case
$d_o = d_s^*$ both marginal deterministic common-category problems
coincide on $\phi^*(T)$, but equality with the unrestricted Shannon
rate-distortion functions there still requires that the finite-$K$
common-category encoder be Shannon-RDF-optimal, not merely Lloyd-Max
stationary. When the verification-side
mapping is monotone increasing in rate, as it is under aggregate
testing with pool-size-driven detection (\Cref{prop:opposing}), the
observation-side mapping is not a standard rate-distortion function,
the \Cref{cor:no-gap} bound does not apply, and the constrained-design
band of \Cref{cor:rate-gap} replaces it with an interval of width
$R_{\max} - R_{\min} = \log_2(K_{\max}/K_{\min})$ bits when non-empty
and with joint infeasibility when empty. The contribution of this paper is therefore
not a tightening of SK within its stated scope. It is the
specialisation of SK to designed-source mapping inside that scope
(\Cref{cor:no-gap}) together with the identification of the
structurally distinct regime that lies outside that scope and the
characterisation of the cardinality-level outer band there (\Cref{cor:rate-gap}).

\section{Comparison with Liu, Shao, Zhang and Poor (2022)}\label{sec:lzp}

Liu, Shao, Zhang and Poor~\cite{lszp2022} give an indirect
rate-distortion characterisation for a semantic source with an
intrinsic state and an extrinsic observation. Their model carries two
distortion measures, one between the intrinsic state and its
reproduction and one between the extrinsic observation and its
reproduction, and the semantic rate-distortion function is the solution
of a convex programming problem in an error covariance matrix. For the
Gaussian extrinsic observation under a linear state-observation
relationship and a diagonalisability condition, they give a
reverse-water-filling solution.

The structural relation to the present paper has two parts. (i) Their
Gaussian semantic rate-distortion function reduces, under the
designed-source specialisation $x = \phi^*(z)$, to standard
rate-distortion on $\phi^*(T)$ in the single-distortion welfare problem.
This is consistent with \Cref{prop:task-loss,prop:decoder,prop:encoder}.
In the scalar Gaussian instance of \Cref{sec:worked-example}, the
specialisation is concrete. The LSZP reverse-water-filling solution
collapses to the Shannon rate-distortion function
$R(D) = \tfrac{1}{2}\log_2(\sigma_\phi^2/D)$ on
$\phi^*(T) \sim \mathcal{N}(0, \sigma_\phi^2)$, which is the rate
function~\eqref{eq:rdf} entering $R_{\min}^{\mathrm{Sh}}$ in~\eqref{eq:R-min}.
(ii) Their framework, like SK's, does not surface the
opposing-monotonicity feasibility band of \Cref{prop:opposing}. The
reason is structural, not a difference in the number of distortion
measures. The LSZP second distortion is on the extrinsic observation,
where adding rate reduces both distortion components, with no
aggregate-detection test where pooling matters.
\Cref{prop:opposing} requires a verification side-constraint in which
the rate variable controls pool size, so finer categories degrade the
aggregate test. This is the structural feature that designed-source
dual-purpose signalling exposes and that LSZP's intrinsic-extrinsic
indirect rate-distortion framework does not capture.

\paragraph{Locating agreement and divergence via \Cref{cor:no-gap}.}
The structural relation above can be made formal. In the standard
regime where the verification-side mapping is monotone non-increasing
in rate, \Cref{cor:no-gap} lower-bounds the common-category SK
rate~\eqref{eq:sk-rd} by $\max\{R(\phi^*(T); D_s),\, R(T; D_o)\}$ and
attains it under compatible reconstructions, and the LSZP indirect
characterisation, in the scalar Gaussian instance of
\Cref{sec:worked-example}, reduces by (i) to the same pair of
Shannon rate-distortion functions on $\phi^*(T)$ and on $T$
respectively. The three characterisations (SK, LSZP, designed-source
reduction) therefore agree on the Shannon-RDF lower-bound structure
in this regime, and agree exactly when the compatible common-category
reconstruction is also RDF-optimal. The present
reduction does not claim a tightening of LSZP within its stated
scope. The divergence considered here arises in the regime in which
the verification-side mapping is not monotone non-increasing in rate,
where welfare and detection are driven in opposite directions by
the same rate variable. Neither SK nor LSZP surfaces the
feasibility band of \Cref{prop:opposing} there, for the structural
reason recorded in (ii).

\section{Numerical worked example: smart-grid dispatch with NTL contrast}\label{sec:numerical}

This section instantiates the reduction on a comms-side worked example.
\Cref{sec:smart-grid} runs the designed-source side on smart-grid
economic dispatch and illustrates that the two SK distortions
trace a single curve up to the
\Cref{prop:task-loss} sandwich. \Cref{sec:ntl-contrast} runs an
unstructured-source contrast on a non-technical-loss (NTL)
detection problem~\cite{leite2018}, where no oracle action exists and
the $(D_s, D_o)$ frontier does not reduce in this way. Both experiments use
numpy/scipy with Lloyd--Max via $K$-means restarts and Nelder--Mead
partition optimisation; the experiment source and JSON-archived
results will be made available upon publication.
Every empirical curve below carries a 95\% percentile confidence
interval from a 200-iteration parametric bootstrap that redraws the
source per iteration and re-runs the full pipeline (codebook fit plus
$\varepsilon(K)$, $\Delta(K)$, or $D_s, D_o$ as appropriate). The
shaded ribbons and error bars on each figure are these CIs.

The dispatch-plus-fraud-detection pairing is not contrived. In a
smart-grid distribution feeder, the meter-level consumption signal is
a single dual-purpose telemetry channel: it feeds the dispatch
optimiser on the designed-source side and the NTL detector on the
verification side at the same time, and the operator's choice of how
many categories $K$ to provision is forced to balance both constraints
against one shared signal. This is exactly the structure
\Cref{prop:opposing} characterises.

\subsection{Smart-grid economic dispatch}\label{sec:smart-grid}

A grid operator dispatches across $N = 4$ time slots against
stochastic per-slot loads $t \in \mathbb{R}^4$, with utility
\[
U(t, r) \;=\; -\tfrac{1}{2}\sum_{i=1}^4 a_i \,(r_i - t_i)^2.
\]
The tracking weights $a_i$ are drawn once from $\mathrm{Uniform}[0.5,
2.0]$ to make the curvature sandwich \Cref{prop:task-loss} non-trivial;
the realised draw gives $\alpha = 0.764$ and $\beta_U = 1.907$, so
$[\alpha/2,\,\beta_U/2] = [0.382,\,0.953]$. The diagonal weight matrix
$W = \mathrm{diag}(a_i)$ is non-isotropic ($\alpha < \beta_U$), so the
Euclidean sandwich~\eqref{eq:sandwich} is not tight on this instance.
The weights are nonetheless fixed across samples, so $W$ is constant
and the conditional-mean decoder of \Cref{prop:decoder} remains exact
for the weighted-quadratic welfare loss. \Cref{fig:smart-grid}\,(a)
records the empirical $\varepsilon(K)$ vs $\Delta(K)$ ratio under this
exact-welfare decoder. The oracle action is
$\phi^*(t) = t$ and welfare is normalised so $U(t, \phi^*(t)) = 0$. We
draw $n_{\mathrm{samples}} = 2 \times 10^4$ truncated-normal load
vectors and, for each $K \in \{2, 4, 8, 16, 32, 64, 128\}$, train the
vector Lloyd--Max quantiser of \Cref{prop:encoder} on $\phi^*(t)$
(\texttt{scipy.cluster.vq.kmeans2}, 20 random restarts). For each $K$
we record the semantic distortion $\varepsilon(K)$ and the welfare loss
$\Delta(K)$.

\Cref{fig:smart-grid}\,(a) plots $\Delta(K)$ against $\varepsilon(K)$
on log--log axes. The measured operating points fall inside the
sandwich envelope at every $K$. The empirical ratio
$\Delta(K) / \varepsilon(K)$ ranges from $0.635$ at $K = 4$ to $0.687$
at $K = 2$, well within $[0.382, 0.953]$. The two SK distortions are
therefore proportional up to a known constant, and the rate-distortion
surface $R(D_s, D_o)$ degenerates to a single curve $R(\varepsilon)$.

\Cref{fig:smart-grid}\,(b) instantiates the feasibility band of
\Cref{prop:opposing} on the same dispatch instance under a meter-level
fraud-detection side-constraint with $M = 200$ meters across the
population (so $m_k = M/K$ per category),
$\sigma_{\mathrm{temp}}^2 = 0.02$, $\sigma_{\mathrm{indiv}}^2 = 1.0$,
effect size $\delta = 1$, false-alarm level
$\alpha_{\mathrm{FA}} = 0.05$, and target power $\beta^* = 0.80$. The
empirical per-coordinate rms dispersion of $\phi^*$ is
$\sigma_\phi = 0.390$ at $d = N = 4$ (bootstrap 95\% CI
$[0.389, 0.393]$), so the vector Gaussian Shannon
rate-distortion function~\cite[Theorem~10.3.3]{coverthomas2006} gives
the Gaussian benchmark
$R_{\min}^{\mathrm{G}}(\varepsilon^*) := R_{\mathrm{Sh}}^{\mathrm{G}}(\varepsilon^*) =
\tfrac{d}{2}\log_2(d\,\sigma_\phi^2/\varepsilon^*)$, equivalently
$K_{\min}^{\mathrm{G}}(\varepsilon^*) = (d\,\sigma_\phi^2/\varepsilon^*)^{d/2}$
under the uniform-assignment correspondence $R_H = R_K = \log_2 K$.
The Cover--Thomas expression is the rate-distortion function of a
Gaussian source. Since the smart-grid load is truncated normal and the
Gaussian maximises differential entropy at fixed variance,
$R_{\min}^{\mathrm{G}}(\varepsilon^*)$ is a tractable upper bound on the
source's Shannon rate-distortion floor. It is therefore a sufficient
Shannon-coding benchmark, not a certificate that a finite-$K$
deterministic Lloyd-Max partition attains the operational endpoint
$R_{\min}(\varepsilon^*)$ of \Cref{cor:rate-gap}. The
$d = 1$ welfare-side expressions of \Cref{sec:worked-example} are
specialisations of these. Setting $\varepsilon^* = 0.20$ yields
$R_{\min}^{\mathrm{G}} \approx 3.21$~bits and $R_{\max} = \log_2 K_{\max} \approx
4.82$~bits, so $K_{\min}^{\mathrm{G}} \approx 9.24$ (bootstrap 95\% CI
$[9.17, 9.52]$) is less than $K_{\max} \approx 28.3$ and the band is
non-empty (shaded). Tightening the welfare target to
$\varepsilon^* = 0.05$ moves $R_{\min}^{\mathrm{G}}$ to $7.21$~bits
($K_{\min}^{\mathrm{G}} \approx 148.0 > K_{\max}$) and the band collapses. The numerical
picture qualitatively tracks the analytical band of
\Cref{sec:worked-example} on a non-Gaussian designed source,
illustrating that the band structure is generic to designed mechanisms
with verification side-constraints and not an artefact of the scalar
Gaussian instance.

\begin{figure}[t]
\centering
\includegraphics[width=\linewidth]{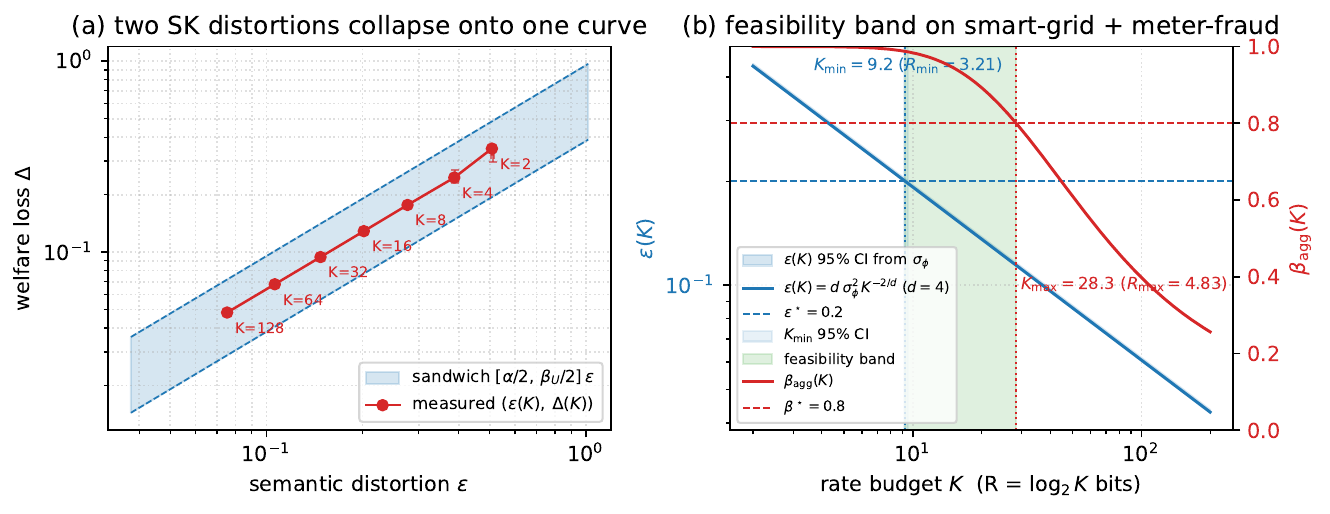}
\caption{Smart-grid economic dispatch worked example. (a) The two SK
distortions $(\varepsilon, \Delta)$ trace a single curve at every $K$,
with the empirical ratio inside the \Cref{prop:task-loss} sandwich
envelope $[\alpha/2,\, \beta_U/2]$. (b) The feasibility band of
\Cref{prop:opposing} on the same dispatch instance under a meter-level
fraud-detection side-constraint, with $\beta_{\mathrm{agg}}(K)$
computed under the uniform-assignment envelope $m_k = M/K$. Vertical
dotted lines locate $K_{\min}^{\mathrm{G}} \approx 9.24$ and
$K_{\max} \approx 28.3$ at $\varepsilon^* = 0.20$.}
\label{fig:smart-grid}
\end{figure}

\Cref{tab:pool-mass} records the realised Lloyd--Max partition statistics
on the same dispatch instance, quantifying the
\Cref{rem:pool-vector} gap between the cardinality-rate envelope
$R_K = \log_2 K$ used in panel~(b) and the operational quantities
$H(C)$ and $m_{\min}$ on the realised partition at $M = 200$. The
entropy gap $\log_2 K - H(C)$ is below $0.085$~bits through $K = 128$,
so $R_H$ tracks $R_K$ closely on this instance and the uniform-assignment
correspondence $R_H = R_K$ used in \Cref{fig:smart-grid}\,(b) is a tight
envelope for entropy. The pool-mass gap is sharper: the deterministic
floor proxy $m_{\min} = \lfloor M \min_k p_k \rfloor$ falls below the
uniform envelope $M/K$ from $K = 16$ onward, and reaches zero at
$K = 128$, meaning that the smallest empirical category has expected
mass below one candidate at $M = 200$. The detection-side endpoint
$\beta_{\mathrm{agg}}$ is then governed by the weakest realised pool
rather than by partition cardinality alone. The cardinality outer band of
\Cref{fig:smart-grid}\,(b) is therefore not tight on the realised
partition above moderate $K$; an operational tightening of the
detection endpoint would index $\beta_{\mathrm{agg}}$ on
$\min_k m_k$ rather than on $\log_2 K$, at the cost of dependence on
the partition-specific pool vector.

\begin{table}[t]
\caption{Lloyd--Max partition statistics on the smart-grid dispatch
instance, pool size $M = 200$, sample size
$n_{\mathrm{samples}} = 2 \times 10^4$. $\min_k p_k$ and $\max_k p_k$
are the smallest and largest category sample fractions, $H(C)$ is the
empirical category entropy in bits, and
$m_{\min} = \lfloor M \cdot \min_k p_k \rfloor$ is a deterministic
proxy for the smallest pool count at $M = 200$. $M/K$ is the
uniform-assignment envelope on $m_{\min}$, used as the detection-side
cardinality bound in \Cref{fig:smart-grid}\,(b).}
\label{tab:pool-mass}
\centering
\begin{tabular}{@{}rrrrrrr@{}}
\toprule
$K$ & $\min_k p_k$ & $\max_k p_k$ & $H(C)$ & $\log_2 K$ & $m_{\min}$ & $M/K$ \\
\midrule
2   & 0.4904 & 0.5096 & 0.9997 & 1 & 98 & 100.00 \\
4   & 0.2389 & 0.2601 & 1.9992 & 2 & 47 & 50.00 \\
8   & 0.1210 & 0.1285 & 2.9997 & 3 & 24 & 25.00 \\
16  & 0.0540 & 0.0906 & 3.9880 & 4 & 10 & 12.50 \\
32  & 0.0204 & 0.0566 & 4.9593 & 5 & 4  & 6.25 \\
64  & 0.0086 & 0.0270 & 5.9344 & 6 & 1  & 3.13 \\
128 & 0.0030 & 0.0162 & 6.9166 & 7 & 0  & 1.56 \\
\bottomrule
\end{tabular}
\end{table}

\subsection{NTL-detection contrast}\label{sec:ntl-contrast}

The reduction relies on the determinism $x = \phi^*(z)$. We confirm
that it does \emph{not} extend to unstructured-source problems by
running a contrast on the NTL setting of Leite and
Mantovani~\cite{leite2018}, where each meter is legitimate with
probability $1 - \pi$ and fraudulent with probability $\pi$, the fraud
indicator $x \in \{0, 1\}$ is latent, and only the noisy meter reading
$z \in \mathbb{R}$ is observed. We use synthetic Gaussian-shift data,
$z \mid x = 0 \sim \mathcal{N}(\mu_L,\, \sigma^2)$ and $z \mid x = 1
\sim \mathcal{N}(\mu_F,\, \sigma^2)$, with
$(\mu_L, \mu_F, \sigma, \pi) = (1.0,\, 0.0,\, 0.5,\, 0.10)$. The
closed-form Bayes risk is $0.0701$. Two SK-style distortions are
evaluated on $K$-interval partitions of $z$ with per-cell centroid
decoder $\hat{z}(c) = \mathbb{E}[z \mid c]$ and per-cell MAP fraud
verdict $\hat{x}(c) = \arg\max_y \Pr(x = y \mid c)$:
\[
D_o \;=\; \mathbb{E}\!\bigl[(z - \hat{z}(c))^2\bigr], \qquad
D_s \;=\; \Pr\bigl[x \neq \hat{x}(c)\bigr].
\]
There is no oracle action: $x$ is genuinely stochastic conditional on
$z$ because the Gaussian-shift channel carries irreducible Bayes risk.
For each $K \in \{4, 8, 16\}$ we sweep a Pareto weight $\lambda \in
[0, 1]$ and optimise the cuts under
$\lambda \tilde{D}_o + (1 - \lambda) D_s$ (with $\tilde{D}_o$ scaled by
the quantile-cut $D_o$ for numerical conditioning), using
\texttt{scipy.optimize.minimize} with Nelder--Mead and warm-started
cuts.

\Cref{fig:ntl} plots the resulting $(D_o, D_s)$ Pareto frontiers. The
three frontiers do not overlap and do not collapse: at $K = 8$, $D_o$
varies from $0.012$ to $0.019$ as $\lambda$ sweeps, while $D_s$ varies
from $0.071$ to $0.078$, and at $K = 16$ the $D_o$ range is
$[0.003,\, 0.009]$ with $D_s$ pinned near the Bayes lower bound. The
$D_s$ achievable from any partition is bounded below by the Bayes
risk, and that bound is reached only when the cuts are aligned with
the likelihood-ratio threshold rather than with the centroid condition
that minimises $D_o$. The cuts that minimise $D_o$ and the cuts that
minimise $D_s$ are distinct at every $K$, so the joint
rate-distortion surface $R(D_s, D_o)$ does not factor through either
distortion alone. SK's two-distortion machinery, the
exponential-tilting parametric decoder, and the generalised
Blahut--Arimoto iteration are the appropriate tools in this regime.
The four-proposition reduction of this paper does not apply.

\begin{figure}[t]
\centering
\includegraphics[width=0.75\linewidth]{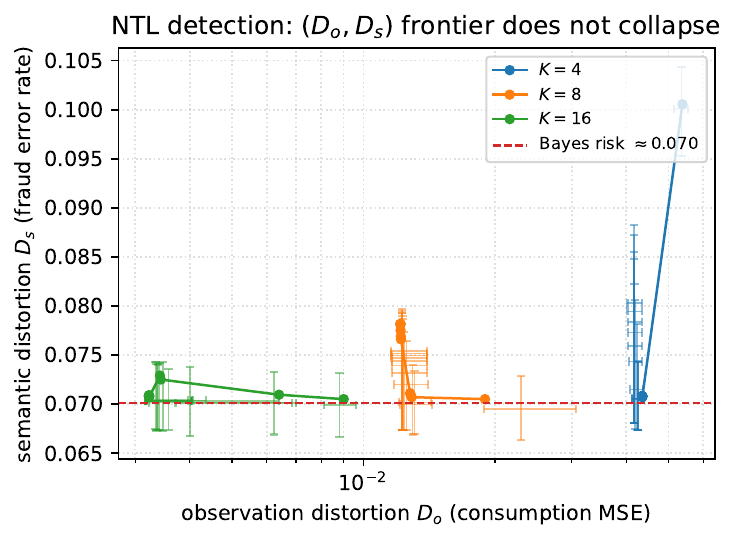}
\caption{NTL detection on a Gaussian-shift unstructured source. For
each $K \in \{4, 8, 16\}$, the $(D_o, D_s)$ Pareto frontier is traced
by sweeping the partition under a weighted objective. The three
frontiers are distinct and non-overlapping, and at every $K$ the
$D_o$-minimal partition differs from the $D_s$-minimal partition.
The dashed horizontal line marks the closed-form Bayes risk
$\approx 0.070$. Per-point error bars are 95\% percentile CIs from
200 parametric bootstrap iterations.}
\label{fig:ntl}
\end{figure}

\section{Discussion}\label{sec:discussion}

\Cref{prop:task-loss,prop:decoder,prop:encoder,prop:opposing} make the
exponential-tilting decoder and the generalised Blahut--Arimoto
algorithm of~\cite{sk2023} unnecessary for welfare-side design within
the designed-source subclass, and \Cref{cor:no-gap} shows that SK's
Markov-chain hypotheses follow from the designed-source variable
mapping rather than needing to be assumed.

The reduction holds because the designed-source class specialises the
SK setting in one structural way: the semantic source $x = \phi^*(z)$
is a deterministic function chosen by the designer, not a stochastic
quantity given by nature. This makes the squared-oracle-error
decoder the conditional mean, exact for quadratic utility (a
definition), the squared-oracle-error encoder Lloyd-Max (a definition),
and the task loss a quadratic sandwich (Taylor's theorem with a zero
linear term).

\Cref{prop:opposing} identifies an additional feature that emerges from
the specialisation. In the standard regime where both distortions
decrease with rate, \Cref{cor:no-gap} lower-bounds the common-category
SK rate by $\max\{R(\phi^*(T); D_s),\, R(T; D_o)\}$, a max of two
standard Shannon rate-distortion functions, and attains it under
compatible common-category reconstructions; the squared-oracle-error
case used in the smart-grid illustration realises compatibility at
the restricted deterministic common-category level, while equality
with the unrestricted Shannon RDFs additionally requires RDF
optimality. For designed mechanisms whose system-level
partition is load-bearing for both allocation and aggregate
verification~\cite{shlezinger2019,strinati2021}, welfare improves with
finer categories while detection degrades, so the two design targets
oppose each other in the rate variable. The joint problem is then not
a standard two-distortion rate-distortion problem: the pool-size
coupling enters as an external constraint on the admissible-encoder
class, and the cardinality-level outer band is the constrained-design
band~\eqref{eq:band} rather than a maximum. This does not
contradict~\cite{sk2023} within its stated stochastic-source scope;
rather, it identifies a regime where the designed-source-mapping
pre-image of the SK setting yields a strictly different joint
structure. Opposing monotonicity arises naturally in any modular
system that commits to a single categorization for both coordination
and verification.
\Cref{cor:no-gap,cor:rate-gap}
together with \Cref{prop:opposing} sharpen this picture: when the
verification-side mapping is monotone decreasing in rate, the SK
characterisation specialises by \Cref{cor:no-gap} and the three
frameworks (SK, LSZP, designed-source reduction) agree on the
Shannon-RDF lower-bound structure, with exact agreement under
RDF-optimal common-category reconstructions; when
it is monotone increasing, the joint problem sits outside the SK
admissible class and the cardinality-level outer band is the
constrained-design band of \Cref{cor:rate-gap}.

By \Cref{prop:task-loss,prop:decoder,prop:encoder}, the
resource-allocation and power-scheduling examples used to motivate
goal-oriented quantization~\cite{zou2023,sun2024} are designed
mechanisms whose single-distortion optimisation, restricted to the
deterministic common-category encoder class of \Cref{sec:mapping},
reduces to Lloyd-Max on the oracle allocation. A designer in the same family of problems, for instance a
grid operator running economic dispatch over a distribution network who
additionally requires a verification constraint such as meter-level
fraud detection, would naturally reach for the SK dual-distortion
framework and inherit its full machinery. The reduction above shows
that this designer need not solve the exponential-tilting or
Blahut-Arimoto problems at all, and that the dual constraint takes the
form of a feasibility band (\Cref{prop:opposing}) rather than
$\max\{R(D_s), R(D_o)\}$.

\subsection{Relation to information bottleneck and indirect
rate-distortion}\label{sec:ib-relation}

The information bottleneck of Tishby, Pereira and
Bialek~\cite{tishby1999ib} compresses a source $X$ to a representation
$T$ that preserves information $I(T; Y)$ about a relevance variable
$Y$. The designed-source setting maps to IB by identifying $X = t$,
$T = c(t)$, $Y = \phi^*(t)$, with one structural specialisation:
$Y = \phi^*(X)$ is a deterministic function of $X$ rather than
stochastic. The deterministic identity replaces the entropy-based IB
relevance criterion with the $L^2$ within-category-variance criterion
$\varepsilon$ on $\phi^*(t)$ that \Cref{prop:encoder} minimises by
Lloyd-Max and that \Cref{prop:task-loss} maps welfare loss into. The
deterministic IB of Strouse and Schwab~\cite{strouse2017dib} replaces
$I(X; T)$ with $H(T)$ to enforce deterministic encoders, and in a
parameter limit of their model DIB clustering becomes equivalent to
$K$-means~\cite{strouse2019ib}. \Cref{prop:encoder}'s Lloyd-Max
encoder instantiates this $K$-means limit on the transformed samples
$\{\phi^*(t_i)\}$ rather than on the raw types $\{t_i\}$, and the
IB-deep-learning phase-transition behaviour of~\cite{tishby2015ib}
(see the IB-machine-learning survey~\cite{goldfeld2020ib}) does not
arise because the compression-relevance curve degenerates to the
single-knob Shannon rate-distortion curve~\eqref{eq:rdf} at every
$K$.

The indirect rate-distortion line is the closer
rate-distortion-theoretic neighbour. The classical formulations of
Dobrushin and Tsybakov~\cite{dobrushin1962} and of Wolf and
Ziv~\cite{wolf1970} encode a hidden source observed through noise;
Witsenhausen~\cite{witsenhausen1980} generalises these by allowing an
arbitrary fidelity criterion between the source and the reconstruction,
and provides the source-coding construction underlying
\Cref{lem:indirect}. The semantic-source variant of Liu, Shao,
Zhang and Poor~\cite{lszp2022} is discussed in \Cref{sec:lzp}.
\Cref{cor:no-gap} lower-bounds the common-category SK rate by the max
of two standard Shannon rate-distortion functions on $\phi^*(T)$ and
on $T$, with equality under compatible common-category reconstructions,
and pins down the conditions under which rate-agreement across SK,
LSZP, and the present reduction holds in that regime. The contribution of the present paper sits on both sides of
the monotonicity axis: \Cref{cor:no-gap} specialises SK inside its
stated scope, and \Cref{prop:opposing,cor:rate-gap} characterise the
feasibility band in the regime outside it.

\subsection{Scope of the reduction}\label{sec:scope}

The reduction is specific to the designed-source subclass. It does
\emph{not} apply to the following adjacent settings.

\begin{itemize}
\item \textbf{General stochastic-source SK.} When the semantic source
  is not directly observable and must be inferred from noisy
  observations~\cite{sk2023}, the full SK machinery is needed: there
  is no oracle action $\phi^*(z)$ to take expectations against, the
  decoder is not a conditional mean of a deterministic function of
  $z$, and Lloyd-Max on $\phi^*$ is undefined. The four propositions
  presented here all turn on the deterministic identity
  $x = \phi^*(z)$.
\item \textbf{Adversarial or strategic SK variants.} When the agents
  whose types $t$ generate the observations have incentives that
  depend on the encoding scheme, the oracle allocation may not be
  designer-implementable without further mechanism-design conditions.
  This layer is out of scope for the present paper.
\item \textbf{Non-smooth or non-concave utility.} A1 and A2 fail when
  $U(t, \cdot)$ is non-smooth (e.g.\ piecewise-linear allocations
  with hard boundaries) or non-concave (e.g.\ multi-modal utility
  arising in non-convex resource problems). The curvature sandwich
  in \Cref{prop:task-loss} no longer applies.
\item \textbf{Non-Euclidean codomain.} The MMSE-based decoder
  \Cref{prop:decoder} and the Lloyd-Max encoder \Cref{prop:encoder}
  both rely on a Euclidean codomain $\mathbb{R}^r$ for the
  allocation. Mechanisms whose actions live in discrete, ordinal, or
  manifold-valued spaces require a different reduction or no
  reduction at all.
\end{itemize}

The contribution of this paper is therefore narrow by construction.
The designed-source subclass is a practically important subclass: it
contains the resource-allocation and power-scheduling problems used to
motivate goal-oriented quantization~\cite{zou2023,sun2024}, and any
designer-chosen allocation rule under smooth concave utility. Within
this subclass, the SK apparatus specialises to the short statements
of \Cref{prop:task-loss,prop:decoder,prop:encoder}, and the additional
feasibility-band structure becomes visible. Outside this subclass, the
SK framework and its full machinery remain necessary.

\bibliographystyle{IEEEtran}
\bibliography{refs}

\end{document}